\documentclass[lettersize,journal]{IEEEtran}
\usepackage{amsmath,amsfonts,amssymb,mathrsfs,bm}
\usepackage{algorithm}
\usepackage{algorithmicx}
\usepackage{algpseudocode}
\usepackage{textcomp}
\usepackage{stfloats}
\usepackage{graphicx}
\usepackage{cite}

\usepackage{xcolor}
\usepackage{colortbl}
\usepackage{makecell}
\usepackage{subfigure}
\usepackage{braket}

\usepackage{amsthm}

\usepackage{multirow}
\usepackage{multicol}
\usepackage{mathtools}

\DeclareMathOperator*{\argmax}{arg\,max}

\newtheorem{proposition}{Proposition}
\newtheorem{remark}{Remark}

\newcommand{\B}{\mathbf}
\newcommand{\BS}{\boldsymbol}
\newcommand{\T}{^\mathsf{T}}
\newcommand{\CT}{^\mathsf{H}}

\begin{document}
	
\title{Near-Field Secure Beamfocusing With Receiver-Centered Protected Zone}

\author{Cen Liu,~\IEEEmembership{Graduate Student Member,~IEEE,}
	Xiangyun Zhou,~\IEEEmembership{Fellow,~IEEE,}
	Nan Yang,~\IEEEmembership{Senior Member,~IEEE,}
	Salman Durrani,~\IEEEmembership{Senior Member,~IEEE,} and
	A. Lee Swindlehurst,~\IEEEmembership{Fellow,~IEEE}
	\thanks{
		This work was supported by the Australian Research Council’s Discovery Projects (DP220101318). An earlier version of this article was presented in part at the IEEE International Conference on Communications, Montreal, Canada, 8--12 June 2025~\cite{CenLiu25}.
		
		C. Liu, X. Zhou, N. Yang, and S. Durrani are with the School of Engineering, The Australian National University, Canberra, ACT 2601, Australia (e-mail: \{cen.liu; xiangyun.zhou; nan.yang; salman.durrani\}@anu.edu.au).
	
		A. L. Swindlehurst is with the Center for Pervasive Communications and Computing, Henry Samueli School of Engineering, University of California, Irvine, CA 92697, USA (e-mail: swindle@uci.edu).}
}



\maketitle

\begin{abstract}

This work studies near-field secure communications through transmit beamfocusing. We examine the benefit of having a protected eavesdropper-free zone around the legitimate receiver, and we determine the worst-case secrecy performance against a potential eavesdropper located anywhere outside the protected zone. A max-min optimization problem is formulated for the beamfocusing design with and without artificial noise transmission. Despite the NP-hardness of the problem, we develop a synchronous gradient descent-ascent framework that approximates the global maximin solution. A low-complexity solution is also derived that delivers excellent performance over a wide range of operating conditions. We further extend this study to a scenario where it is not possible to physically enforce a protected zone. To this end, we consider secure communications through the creation of a virtual protected zone using a full-duplex legitimate receiver. Numerical results demonstrate that exploiting either the physical or virtual receiver-centered protected zone with appropriately designed beamfocusing is an effective strategy for achieving secure near-field communications.

\end{abstract}


\section{Introduction}

\IEEEPARstart{N}{ear-field} wireless communications have emerged as a pivotal technological paradigm for next-generation wireless networks, attracting growing interest from both academia and industry~\cite{YuanweiLiu24, YuanweiLiu23, MingyaoCui23, HaiyangZhang23}. With the deployment of large-scale antenna arrays and the exploitation of high-frequency spectrum resources, near-field communication zones can extend from tens to hundreds of meters. Within this range, the spherical wavefronts of near-field signals provide an additional dimension that can be exploited for location-based communication, presenting a transformative shift for wireless systems. This shift particularly calls for novel physical-layer security designs that move beyond traditional methods relying on far-field propagation~\cite{XiangyunZhou13, NanYang15CM, Mukherjee14}.

\subsection{Motivation and Related Work}

Techniques for physical-layer security in traditional far-field communication systems have considered the concept of a protected zone \cite{XiangyunZhou11, Romero13}, in which an eavesdropper-free area is assumed to exist in a region surrounding the transmitter. The benefit of such a transmitter-side protected zone is that the received signal-to-noise ratio (SNR) at the eavesdropper is degraded due to the distance-dependent path loss between the transmitter and eavesdropper. Besides, from a practical engineering perspective, it is often more convenient to create a protected zone around a base station or an access point (i.e., the transmitter in the downlink) since their location is typically fixed and access to them is controlled. This explains why prior work has almost always assumed that the protected zone surrounds the transmitter rather than the receiver. One exception is found in~\cite{Yavuz21} where an aerial transmitter serves ground users in the presence of ground eavesdroppers. In this specific system setup, the eavesdroppers were constrained to be on the ground, and hence could not be located close to the aerial transmitter. This constrained scenario motivated~\cite{Yavuz21} to investigate the benefit and design of the protected zone around the legitimate receivers with far-field beamforming from the transmitter.

On the other hand, the inherent \textit{beamfocusing} capability of near-field communication systems provides increased degrees of freedom (DoFs) for location-based physical-layer security designs. This observation has been exploited by a number of researchers. For example, the authors in~\cite{Anaya22} found that spherical wave propagation enables the distance of the eavesdropper from the receiver to serve as a new spatial DoF for secure communications. The work of \cite{ZhengZhang24} further revealed that positive secrecy capacities are achievable within the near-field region even when the legitimate receiver and eavesdropper are aligned in the same direction and the eavesdropper is closer to the transmitter than the legitimate receiver. Near-field hybrid beamfocusing design was investigated in \cite{Nasir24} for multiple legitimate receivers in the presence of multiple non-colluding eavesdroppers. Based on explicit electromagnetic modeling, \cite{Droulias24} analytically studied the superiority of focused beams in enhancing secrecy performance with respect to beams generated by far-field beamforming. True-time-delay-based analog beamforming for near-field wideband systems was considered in \cite{YuchenZhang24} to mitigate the beam-split effect. The authors in~\cite{BoqunZhao24} analyzed the trade-off between secrecy performance and power control for commonly-employed far- and near-field channel models. The results in \cite{YunpuZhang25} and~\cite{ZhifengTang25} showed notable secrecy gains when using artificial noise (AN) in conjunction with beamfocusing. It was also demonstrated in \cite{Ferreira24} that receive beamfocusing effectively filters out jamming signals in the uplink of a near-field communication system.

The above studies have assumed that perfect channel state information (CSI) for the eavesdropper is available at the transmitter. Robust beamfocusing designs with AN transmission from either the legitimate transmitter or receiver have been investigated in~\cite{WeijianChen25}. In~\cite{JiangongChen24}, near-field directional modulation was designed based only on the legitimate receiver's direction and range. The results in~\cite{LinlinSun25} demonstrated that a random frequency diverse array can effectively reduce the received signal strength at locations other than that of the legitimate receiver, while \cite{ZihaoTeng25} presented a dynamic symbol-level precoder that can distort the signal constellations perceived at these unintended locations. Nevertheless, none of these works considered beamfocusing design for maximizing the worst-case secrecy performance in which the eavesdropper can optimally choose its eavesdropping location.

A key insight from the prior work above is that the secrecy performance of near-field communication systems with large-scale antenna arrays is primarily influenced by the spatial disparity between the legitimate receiver and the eavesdropper. Thus, in the context of near-field physical-layer security, the optimal approach for the eavesdropper is to position itself as close as possible to the legitimate receiver so as to maximize its received SNR, since with range selectivity the transmitter can reduce the SNR when the eavesdropper is nearby. We note that~\cite{Anaya22} briefly discussed the potential benefits of delimiting an eavesdropper-free area close to the legitimate receiver; however, the design of a secure near-field communication system exploiting such a protected zone has not yet been studied, which motivates this research work.

\begin{table*}[t]
	\centering
	\caption{Comparison of Key Considerations in This Work With Literature on Near-Field Security-Oriented Design}
		\begin{tabular}{l|c|c|c|c|c|c|c|c|c|c}
			\hline
			& \cite{Anaya22} & \cite{ZhengZhang24} & \cite{Nasir24} & \cite{YuchenZhang24} & \cite{YunpuZhang25}& \cite{ZhifengTang25} & \cite{WeijianChen25}& \cite{JiangongChen24} & \cite{ZihaoTeng25} & Ours\\
			
			\rowcolor[gray]{0.9}
			Unavailability of Eavesdropper's CSI for Design & & & & & & & & \checkmark & \checkmark & \checkmark \\
			
			AN Injection from Transmitter & & & & & \checkmark & \checkmark & \checkmark & \checkmark & \checkmark & \checkmark \\
			
			\rowcolor[gray]{0.9}
			AN Injection from Full-Duplex Receiver & & & & & & & \checkmark & & & \checkmark \\
			
			Low-Complexity Design Approach & \checkmark &  &  & \checkmark & \checkmark & \checkmark & & \checkmark & \checkmark & \checkmark \\
			
			\rowcolor[gray]{0.9}
			Security against Multiple Non-Colluding Eavesdroppers & \checkmark & & \checkmark & & & & & & \checkmark & \checkmark \\
			
			Security against Worst Possible Eavesdropper Location & &  &  & & & & & & & \checkmark \\
			\hline
		\end{tabular}
	\label{table_literature}
\end{table*}

\subsection{Contributions}

Motivated by the gap noted above in prior research, in this work we consider the concept of a \textit{Receiver-Centered Protected Zone} (RCPZ) for physical-layer security, which corresponds to an eavesdropper-free area around the legitimate receiver. There are plenty of practical scenarios in which it is reasonable to assume that such a protected zone exists around legitimate receivers in the downlink of commercial or military wireless networks. For instance, designated reception areas with controlled access and physical barriers may be created to prevent eavesdroppers from entering. Moreover, legitimate users may visually inspect the surrounding environment or conduct electronic surveillance to discover and neutralize radio-frequency (RF) intruders. To take advantage of an RCPZ, we design analog beamfocusing at the transmitter as a cost-effective implementation for a large-scale antenna array. Our design aims to achieve secure communication by considering the worst case in the sense that the eavesdropper may be located anywhere outside the RCPZ.

The main contributions of this work are as follows:
\begin{itemize}
	\item To the best of our knowledge, this is the first work on near-field secure communication design that exploits a protected eavesdropper-free zone around the legitimate receiver. An RCPZ-based near-field beamfocusing design enables strong signal reception at the receiver while constraining the quality of signal reception at all possible eavesdropping locations outside the protected zone.
	
	\item We develop cutting-edge beamfocusing designs both with and without AN generated by the transmitter. We formulate a max-min optimization problem to investigate the system's worst-case secrecy performance. Despite the NP-hardness of the nonconcave-nonconvex (NCNC) max-min optimization, we provide an approximate global maximin solution using a proposed synchronous gradient descent-ascent framework.
	
	\item We develop an alternative beamfocusing design with low computational complexity, which is shown to achieve excellent secrecy performance over a wide range of scenarios.
	
	\item We further extend our study to the scenario where a physical RCPZ is no longer possible and, as a result, there is no restriction on the location of the eavesdropper. For this scenario, we consider a full-duplex legitimate receiver capable of simultaneously receiving information signals and transmitting AN signals. The AN signal transmitted by the receiver discourages the eavesdropper from locating itself too close to the receiver, effectively creating a virtual protected zone. We then provide an approximate global maximin beamfocusing solution via an extended synchronous gradient descent-ascent framework. Numerical results demonstrate that the self-interference cancellation (SIC) performance of the full-duplex receiver plays a vital role in achieving effective secrecy performance.
\end{itemize}

For ease of comparison, Table~\ref{table_literature} summarizes the key considerations in this work and other related literature on near-field security-oriented system design.

\subsection{Organization and Notation}

The rest of this paper is organized as follows. Section~\ref{section_system_model} outlines the system model and introduces the concept of the RCPZ. Section~\ref{section_problem_formulation} formulates the max-min secrecy optimization problem using beamfocusing. In Sections~\ref{section_RCPZ_SGDA} and \ref{section_RCPZ_low_complexity}, we propose the synchronous gradient descent-ascent framework to solve the max-min optimization problem and further develop an alternative low-complexity solution. In Section~\ref{section_RCVPZ}, we extend our study to the scenario of a virtual protected zone and propose an extended synchronous gradient descent-ascent framework for this scenario. Section~\ref{section_numerical_results} presents numerical results to evaluate the performance of the proposed designs. Finally, we conclude the paper in Section~\ref{section_conclusion}.

\textit{Notation:} Scalars, vectors and matrices are denoted by italic, bold-face lower-case and bold-face upper-case letters, respectively. Hermitian, transpose, conjugate and gradient operations are denoted by $\left(\cdot\right)\CT$, $\left(\cdot\right)\T$, $\left(\cdot\right)^*$ and $\nabla\left(\cdot\right)$, respectively. $N$-dimensional real- and complex-valued column vector spaces are respectively denoted by $\mathbb{R}^N$ and $\mathbb{C}^N$, and the space of $N\times M$ complex-valued matrices is denoted by $\mathbb{C}^{N\times M}$. The real part and modulus of a complex number are respectively denoted by $\Re\left(\cdot\right)$ and $|\cdot|$, and the Euclidean norm of a complex-valued vector is denoted by $\|\cdot\|$. The distribution of a circularly symmetric complex Gaussian random vector with mean vector $\BS\mu$ and covariance matrix $\BS\Sigma$ is denoted by $\mathcal{CN}\left(\BS\mu,\BS\Sigma\right)$. The identity matrix of size $N$ is denoted by $\B{I}_N$. The relative complement of a set $\mathcal{A}$ in a set $\mathcal{B}$ and their Cartesian product are denoted by $\mathcal{B}\setminus\mathcal{A}$ and $\mathcal{A}\times\mathcal{B}$, respectively.

\section{System Model and Receiver-Centered Protected Zone} \label{section_system_model}

\begin{figure}[t]
	\centerline{\includegraphics[width=0.9\columnwidth]{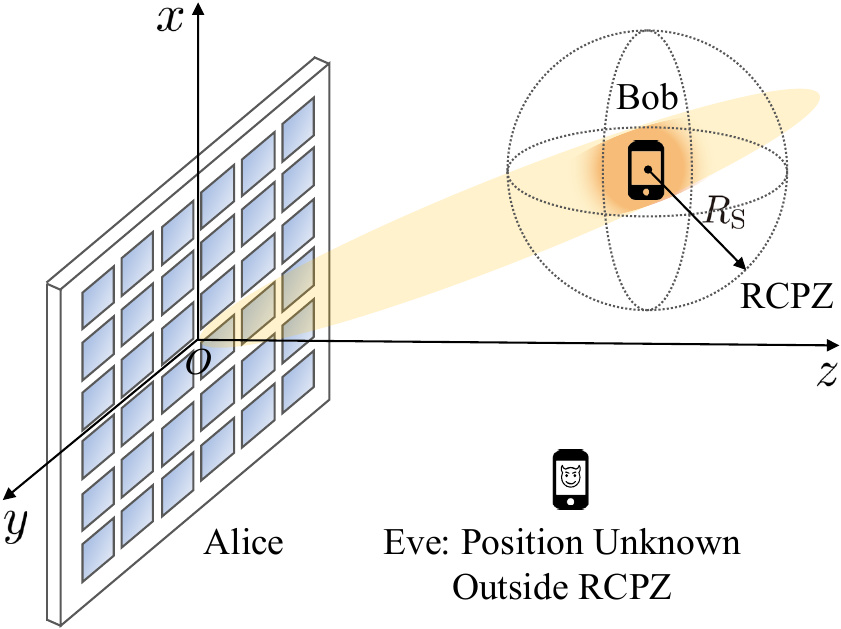}}
	\caption{The considered near-field communication system where Alice is equipped with a UPA, Bob is surrounded by an RCPZ, and Eve is located anywhere outside the RCPZ.}
	\label{fig_model}
\end{figure}
	
We consider a near-field communication system where Alice performs downlink secure transmission to Bob in the presence of an eavesdropper Eve, as illustrated in Fig.~\ref{fig_model}. A large-scale uniform planar array (UPA) is deployed by Alice with a total of $N=N_{\rm x} N_{\rm y}$ antennas, where $N_{\rm x}$ and $N_{\rm y}$ represent the number of antennas along the $x$- and $y$-axis, respectively. The square-shaped antennas, all of the same type with diagonal aperture $D$, are deployed edge-to-edge in the $xy$-plane~\cite{Bjornson21}. The \textit{Rayleigh distance} of each antenna element and \textit{Rayleigh array distance} of the entire array can be computed as $d_{\rm R} = 2D^2/\lambda$ and $d_{\rm RA} = \left(N_{\rm x}^2+N_{\rm y}^2\right)D^2/\lambda$, respectively, with $\lambda$ denoting the carrier wavelength~\cite{Selvan17, Bjornson21}. Let $m\in\left\{1,\ldots,N_{\rm x}\right\}$ and $n\in\left\{1,\ldots,N_{\rm y}\right\}$ respectively represent the row index along the $x$-axis and the column index along the $y$-axis. The $\left(m,n\right)$-th antenna is centered at the point $\B{p}_{\rm A}^{m,n} = \left[x_m,y_n,0\right]\T$, where
\begin{align}
	x_m = \left(m-\frac{N_{\rm x}+1}{2}\right)\frac{D}{\sqrt{2}},\
	y_n = \left(n-\frac{N_{\rm y}+1}{2}\right)\frac{D}{\sqrt{2}}.
\end{align}

A single-antenna legitimate receiver Bob, with a known position $\B{p}_{\rm B} = \left[x_{\rm B},y_{\rm B},z_{\rm B}\right]\T$, is located in the near field of Alice's transmit array, i.e., $d_{\rm AB} = \|\B{p}_{\rm B}\| \leq d_{\rm RA}$. The position of the eavesdropper Eve is denoted by $\B{p}_{\rm E} = \left[x_{\rm E},y_{\rm E},z_{\rm E}\right]\T$. We assume that all users (Alice, Bob and Eve) are located in the far field of the individual antennas, i.e., their distances are larger than the Rayleigh distance of a single antenna element $d_{\rm R} = 2D^2/\lambda$. This is a widely-used and practical assumption implicitly adopted in the literature of near-field communications~\cite{Anaya22, ZhengZhang24, Nasir24, Droulias24, YuchenZhang24, BoqunZhao24, YunpuZhang25, ZhifengTang25, Ferreira24, WeijianChen25, JiangongChen24, LinlinSun25, ZihaoTeng25, Bjornson21}. For example, if we assume a nominal antenna aperture of $\frac{\sqrt{2}}{2}\lambda$ for each of Alice's antennas, the value of $d_{\rm R}$ is approximately $9\, {\rm cm}$ and $1\, {\rm cm}$ with carrier frequencies of $3.5\, {\rm GHz}$ and $28\, {\rm GHz}$, respectively. The line-of-sight (LoS) channel for Bob $\B{h}_{\rm B}\in\mathbb{C}^N$ and Eve $\B{h}_{\rm E}\in\mathbb{C}^N$ are respectively modeled by
\begin{align}
	h_{\rm B}^k = \alpha_{\rm B}^k \exp\left(-j\kappa d_{\rm B}^k\right)
\end{align}
and
\begin{align}
	h_{\rm E}^k = \alpha_{\rm E}^k \exp\left(-j\kappa d_{\rm E}^k\right),
\end{align}
where the channel coefficients are indexed by $k = \left(m-1\right)N_{\rm y} + n \in\left\{1,\ldots,N\right\}$, $\kappa=2\pi/\lambda$ is the wavenumber, $\alpha_{\rm B}^k = (2\kappa d_{\rm B}^k)^{-1}$ and $\alpha_{\rm E}^k = (2\kappa d_{\rm E}^k)^{-1}$ are the free-space path gains of Bob's and Eve's channel with respect to the $k$-th transmit antenna of Alice located at $\B{p}_{\rm A}^k = \B{p}_{\rm A}^{m,n}$, and $d_{\rm B}^k = \left\|\B{p}_{\rm B} - \B{p}_{\rm A}^k\right\|$, $d_{\rm E}^k = \left\|\B{p}_{\rm E} - \B{p}_{\rm A}^k\right\|$ denote the Euclidean distance between Bob/Eve and the $k$-th transmit antenna of Alice, respectively.

The signals received by Bob and Eve are given by
\begin{align}
	r_{\rm B} = \B{h}_{\rm B}\CT \B{x} + n_{\rm B}
\end{align}
and
\begin{align}
	r_{\rm E} = \B{h}_{\rm E}\CT \B{x} + n_{\rm E},
\end{align}
respectively, where $\B{x}\in\mathbb{C}^N$ is the transmitted signal, $n_{\rm B}\sim\mathcal{CN}\left(0,\sigma_{\rm B}^2\right)$ and $n_{\rm E}\sim\mathcal{CN}\left(0,\sigma_{\rm E}^2\right)$ represent additive white Gaussian noise (AWGN). We assume that Bob’s CSI and location are known at Alice by means of channel estimation~\cite{YuanjianLi25} and/or localization techniques~\cite{YueWang18}, and they are also assumed to be known at Eve. In contrast, Eve’s CSI and location are unknown to Alice and Bob.

Considering the lack of CSI for Eve, we adopt a widely-used secure transmission strategy where Alice transmits both the information signal and an AN signal, i.e.,
\begin{align}
	\B{x} = \B{w}s + \B{z},
\end{align}
where $\B{w}\in\mathbb{C}^N$ denotes the near-field analog beamfocusing vector at Alice, $s$ is the information signal with normalized power $\mathbb{E}\left[|s|^2\right]=1$, and $\B{z}\in\mathbb{C}^N$ denotes the AN signal. The AN signal generated by Alice is chosen to lie in the null space of Bob's channel: $\B{h}_{\rm B}\CT\B{z}=0$, which is a widely adopted AN transmission strategy~\cite{Goel08, XiangyunZhou10, NanYang15, YuanjianLi17}. Then, the signals received by Bob and Eve are given by, respectively,
\begin{align}
	r_{\rm B} = \B{h}_{\rm B}\CT \B{w}s + n_{\rm B}
\end{align}
and
\begin{align}
	r_{\rm E} = \B{h}_{\rm E}\CT \B{w}s + \B{h}_{\rm E}\CT\B{z} + n_{\rm E}.
\end{align}
Let $\B{Z}\in\mathbb{C}^{N\times\left(N-1\right)}$ be an orthonormal basis for the null space of $\B{h}_{\rm B}\CT$, so that $\B{Z}\CT\B{Z} = \B{I}_{N-1}$ and $\B{z}=\B{Z}\B{v}$. The elements of $\B{v}\in\mathbb{C}^{N-1}$ are chosen to be i.i.d. complex Gaussian with zero mean and variance $\sigma_v^2$, i.e., $\B{v}\sim\mathcal{CN}\left(\B{0},\sigma_v^2\B{I}_{N-1}\right)$. Defining a power allocation factor $\phi = \|\B{w}\|^2/P_{\rm A}$ as the ratio of the power of the information bearing signal to the total transmit power $P_{\rm A}$ at Alice, the variance of each element of $\B{v}$ is given by $\sigma_v^2 = \left(1-\phi\right)P_{\rm A}/\left(N-1\right)$. Here, we note that the special case of no AN corresponds to $\phi=1$.

In this work, we exploit the availability of an RCPZ, which is a known eavesdropper-free region centered at Bob. The shape of the RCPZ can in general be arbitrary. For ease of analysis, we define a \textit{Security Radius} $R_{\rm S}$, inside of which Bob can ensure (e.g., by inspection, detection or physical obstruction) that no eavesdropper is present. Mathematically, this assumption can be expressed as $\left\|\B{p}_{\rm E}-\B{p}_{\rm B}\right\|\geq R_{\rm S}$. From the perspective of physical-layer security, the RCPZ prevents Eve from positioning herself near Bob and hindering secure communications.

\section{Problem Formulation}\label{section_problem_formulation}

\subsection{Max-Min Optimization for Worst-Case Security} \label{section_max_min_optimization}

The secrecy capacity for transmission from Alice to Bob can be written as $C_{\rm S}=\left(C_{\rm B}-C_{\rm E}\right)^+$ where $\left(\cdot\right)^+ =\max\left\{\cdot,0\right\}$.
\begin{align}
	C_{\rm B} = \log_2\left(1 + \frac{|\B{h}_{\rm B}\CT \B{w}|^2}{\sigma_{\rm B}^2}\right)
\end{align}
and
\begin{align}
	C_{\rm E} = \log_2\left(1 + \frac{|\B{h}_{\rm E}\CT \B{w}|^2}{\frac{\left(1-\phi\right)P_{\rm A}}{N-1} \|\B{h}_{\rm E}\CT\B{Z}\|^2 + \sigma_{\rm E}^2}\right)
\end{align}
are the individual channel capacities for Bob and Eve, respectively. We consider the problem of jointly optimizing the transmit power $P_{\rm A}$, the power allocation factor $\phi$, and the beamfocusing vector $\B{w}$ at Alice to maximize the secrecy capacity, while assuming that Eve can intelligently\footnote{An intelligent Eve is capable of moving to the best eavesdropping location (outside the RCPZ) with the highest signal-to-interference-plus-noise ratio (SINR) based on her knowledge of Alice's design strategy.} choose her location anywhere outside the RCPZ. This represents the worst-case scenario from the legitimate user’s perspective and leads to the following max-min optimization problem:
\begin{subequations}
	\label{max_min_problem}
	\begin{align}
		\text{(P1)}: \max_{P_{\rm A}, \phi, \B{w}}\, &\min_{\B{p}_{\rm E}}\ C_{\rm S} \left(P_{\rm A}, \phi, \B{w}, \B{p}_{\rm E}\right)\\
		{\rm s.t.}\
		& 0 \leq P_{\rm A} \leq P_{\rm TX},\label{power_budget_const}\\
		& 0 \leq \phi \leq 1,\label{alloc_const}\\
		& \|\B{w}\|^2 = \phi P_{\rm A},\label{signal_power_const}\\
		& |w_i| = |w_j|,\ \forall i,j\in\left\{1,2,\ldots,N\right\},\label{umod_const}\\
		& \|\B{p}_{\rm E}-\B{p}_{\rm B}\|\geq R_{\rm S},\label{rcpz_const}\\
		& \|\B{p}_{\rm E}-\B{p}_{\rm A}^k\| \geq d_{\rm R},\ \forall k\in\left\{1,2,\ldots,N\right\},\label{ff_const}
	\end{align}
\end{subequations}
where (\ref{power_budget_const}) limits the transmit power at Alice within a power budget $P_{\rm TX}$, (\ref{alloc_const}) limits the value of the power allocation factor, (\ref{signal_power_const}) sets the transmit power of the information bearing signal, (\ref{umod_const}) incorporates the unit-modulus constraint of the analog beamformer at Alice, (\ref{rcpz_const}) defines the RCPZ, and (\ref{ff_const}) models the assumption that Eve is located in the far field of each of Alice's antenna elements.

The optimization problem in (P1) is an NCNC max-min problem with nonconvex constraints (\ref{signal_power_const})$\sim$(\ref{ff_const}), since the objective is neither a concave function of the outer maximization variables nor a convex function of the inner minimization variable. Moreover, it can be viewed as a two-player sequential game~\cite{Razaviyayn20} between Alice and Eve with the secrecy capacity as the payoff. Alice, the first player, determines the optimal transmission design $\left\{P_{\rm A}^{\rm max}, \phi^{\rm max}, \B{w}^{\rm max}\right\}$ that maximizes the worst-case secrecy capacity $\tilde{C}_{\rm S} =\min_{\B{p}_{\rm E}} C_{\rm S}$ without prior knowledge of Eve's strategy, i.e., her position $\B{p}_{\rm E}$. Given Alice's strategy, Eve aims to minimize the secrecy capacity $C_{\rm S}$ by selecting the best position $\B{p}_{\rm E}^{\rm min}$ for eavesdropping. While the problem in (P1) focuses on one eavesdropper at an unknown location, this work remains valid for scenarios with an arbitrary number of non-colluding eavesdroppers in the sense of the worst-case secrecy performance.

\subsection{Focal-Point-Based Beamfocusing Design}

The first-order gradient descent-ascent (GDA) framework is commonly utilized to iteratively solve max-min optimization problems and converge to a local maximin stationary solution~\cite{ChiJin20}. However, it has been empirically observed that GDA algorithms may diverge when the outer variables do not converge slowly enough~\cite{TianyiLin20} and/or high-dimensional feasible sets are not sufficiently restricted. To avoid this, we propose a dimensionality reduction approach for designing the beamfocusing vector by explicitly designating the position of the beam's focal point (FP) $\B{p}_{\rm F} = \left[x_{\rm F},y_{\rm F},z_{\rm F}\right]\T$. The vector $\B{w}_{\rm F} = \sqrt{\phi P_{\rm A}} \B{\bar h}_{\rm F}$ implements maximum ratio transmission (MRT) analog beamfocusing at the FP, where $\B{\bar{h}}_{\rm F}$ denotes the FP LoS channel whose elements are given by $\bar{h}_{\rm F}^k = \frac{1}{\sqrt{N}} \exp\left(-j\kappa d_{\rm F}^k\right)$, where $d_{\rm F}^k = \left\|\B{p}_{\rm F} - \B{p}_{\rm A}^k\right\|$. Based on the problem formulation in (P1), the FP-based beamfocusing design for the worst-case secrecy capacity can be reformulated as
\begin{subequations}
	\label{max_min_fp_problem_v1}
	\begin{align}
		\text{(P2)}: \max_{P_{\rm A}, \phi, \B{p}_{\rm F}}\, &\min_{\B{p}_{\rm E}}\ C_{\rm S} \left(P_{\rm A}, \phi, \B{p}_{\rm F}, \B{p}_{\rm E}\right)\\
		{\rm s.t.}\
		& \B{p}_{\rm F}\in\mathcal{R}_{\rm B},\label{fp_const}\\
		& {\rm \left(\ref{power_budget_const}\right)}, {\rm \left(\ref{alloc_const}\right)}, {\rm \left(\ref{rcpz_const}\right)}, {\rm \left(\ref{ff_const}\right)},\notag
	\end{align}
\end{subequations}
where the convex constraint (\ref{fp_const}) explicitly restricts the FP position to lie along the ray $\mathcal{R}_{\rm B} = \left\{\B{p}\in\mathbb{R}^3|\B{p}=d_{\rm F}\B{p}_{\rm B}/\|\B{p}_{\rm B}\|,\forall d_{\rm F}>0\right\}$ directed from the origin towards Bob. It is natural to consider (\ref{fp_const}) since even a slight deviation of the FP from $\mathcal{R}_{\rm B}$ can cause a considerable SNR degradation at Bob, given the pencil-like nature of the beams created by the large-scale antenna array. Note that the transmit power constraint (\ref{signal_power_const}) of the information bearing signal and the unit-modulus constraint (\ref{umod_const}) are eliminated from the problem in (P2) because they are already satisfied by the FP-based beamfocusing.

\subsection{Optimal Transmit Power at Alice}

Here, we derive the optimal transmit power at Alice to further simplify the problem in (P2).
\begin{proposition}
	The optimal value for $P_{\rm A}$ in the max-min problem of (P2) is achieved when Alice transmits with full power, i.e.,
	\begin{align}
		P_{\rm A}^{\rm opt} = \argmax_{0 \leq P_{\rm A} \leq P_{\rm TX}} C_{\rm S} \left(P_{\rm A}, \bar\phi, \B{\bar p}_{\rm F}, \B{\bar p}_{\rm E}\right) = P_{\rm TX}
	\end{align}
	for any feasible power allocation factor $\bar\phi$, FP position $\B{\bar p}_{\rm F}$ and location for Eve $\B{\bar p}_{\rm E}$.
	\label{proposition_1}
\end{proposition}

\begin{proof}
	The proof is provided in Appendix~\ref{appendix_A}.
\end{proof}

Based on Proposition~\ref{proposition_1}, the problem formulation in (P2) can be reduced to the following equivalent max-min problem with the optimization variable $P_{\rm A}$ removed:
\begin{subequations}
	\label{max_min_fp_problem_v2}
	\begin{align}
		\text{(P3)}:\ \max_{\phi, \B{p}_{\rm F}}\ &\min_{\B{p}_{\rm E}}\ C_{\rm S} \left(\phi, \B{p}_{\rm F}, \B{p}_{\rm E}\right)\\
		{\rm s.t.}\
		& {\rm \left(\ref{alloc_const}\right)}, {\rm \left(\ref{rcpz_const}\right)}, {\rm \left(\ref{ff_const}\right)}, {\rm \left(\ref{fp_const}\right)}, \notag
	\end{align}
\end{subequations}
where the power budget $P_{\rm TX}$ is assigned to the transmit power $P_{\rm A}$ in the objective function. Nevertheless, (P3) remains an NP-hard NCNC max-min problem due to the infinitely many possible eavesdropping positions outside the RCPZ.

\section{Proposed Sychronous Gradient Descent-Ascent Framework} \label{section_RCPZ_SGDA}

In this section, a synchronous gradient descent-ascent (SGDA) framework is proposed to obtain the approximate global optimal solution to the max-min problem in (P3). The SGDA framework alternates between using the gradient projection method for the outer secrecy capacity maximization subproblem and the augmented Lagrangian method for the inner secrecy capacity minimization subproblem.

\subsection{Optimal Power Allocation Strategy}

Considering a fixed FP $\B{\bar p}_{\rm F}$ and a fixed position for Eve $\B{\bar p}_{\rm E}$, the optimal allocation of Alice's transmit power between the information bearing signal and the AN signal is determined by the following proposition.
\begin{proposition}
	The optimal power allocation factor $\phi^{\rm max}$ for the max-min problem in (P3) is obtained as
	\begin{align}
		\Phi\left(\B{\bar p}_{\rm F},\B{\bar p}_{\rm E}\right) = 
		\left\{
		\begin{aligned}
			&\, \multirow{2}{*}{$\min\left\{\phi_a^+,1\right\}$,} &\gamma>\beta>0 {\rm\ or}\\
			& &0<\gamma<\beta, \Delta>0,\\
			&\min\left\{\phi_b^+,1\right\}, &\gamma=\beta>0,\\
			&H\left( \alpha\sigma_{\rm E}^2 - \beta\sigma_{\rm B}^2 \right), &\gamma=0,\\
			&1, &\gamma>\beta=0,\\
			&0, &0<\gamma<\beta, \Delta\leq0,
		\end{aligned}
		\right.
	\end{align}
	for any feasible position of the FP $\B{\bar p}_{\rm F}$ and Eve $\B{\bar p}_{\rm E}$. Definitions for $\phi_a$, $\phi_b$, $\alpha$, $\beta$, $\gamma$, $\Delta$ and $H\left(\cdot\right)$ are provided in Appendix~\ref{appendix_B}. 
	\label{proposition_2}
\end{proposition}
	
\begin{proof}
	The proof is provided in Appendix~\ref{appendix_B}.
\end{proof}

\subsection{Gradient Projection Method for Outer Maximization} \label{section_GP}

With a fixed position for Eve $\B{\bar p}_{\rm E}$ and a given power allocation factor $\bar\phi$, the FP-based design finds the optimal FP position for secrecy capacity maximization, given by
\begin{align}
	\text{(SP1)}:\ \max_{\B{p}_{\rm F}\in\mathcal{R}_{\rm B}}\ C_{\rm S} \left(\bar\phi, \B{p}_{\rm F}, \B{\bar p}_{\rm E}\right).
\end{align}
Owing to the convexity of the feasible region $\mathcal{R}_{\rm B}$, the gradient projection (GP) method~\cite{GPM} can be applied to iteratively update the FP position. The $l$-th iteration of the GP algorithm can be written as
\begin{align}
	\B{p}_{\rm F}^{l+1} = \BS{P}_{\mathcal{R}_{\rm B}}\left[\B{p}_{\rm F}^l + \beta_{\rm F}^l\nabla_{\B{p}_{\rm F}} C_{\rm S}\left(\bar\phi, \B{p}_{\rm F}^l, \B{\bar p}_{\rm E}\right)\right],
\end{align}
where $\BS{P}_{\mathcal{R}_{\rm B}}\left[\cdot\right]$ denotes the Euclidean projection onto the convex set $\mathcal{R}_{\rm B}$, $\beta_{\rm F}^l$ represents the \textit{Armijo} step size selected by the backtracking line search method~\cite{NumericalOptimization}, and $\nabla_{\B{p}_{\rm F}} C_{\rm S}$ is the gradient of the objective function with respect to the FP position, the expression for which is provided in (\ref{grad_Cs_pF}) in Appendix~\ref{appendix_C}.

\subsection{Augmented Lagrangian Method for Inner Minimization} \label{section_AL}

With a fixed FP $\B{\bar p}_{\rm F}$ and power allocation factor $\bar\phi$, Eve aims to find the best eavesdropping position $\B{p}_{\rm E}^{\rm min}$ for secrecy capacity minimization, given by
\begin{align}
	\text{(SP2)}:\ \min_{\B{p}_{\rm E}}\ C_{\rm S} \left(\bar\phi, \B{\bar p}_{\rm F}, \B{p}_{\rm E}\right),\ {\rm s.t.}{\rm (\ref{rcpz_const}),(\ref{ff_const}).}
\end{align}
Replacing the non-differentiable constraints (\ref{rcpz_const}) and (\ref{ff_const}) with their differentiable counterparts (i.e., taking the square of both sides of the inequality) and introducing the slack variables $\B{s} = \left[s_0,s_1,\ldots,s_N\right]\T\in\mathbb{R}^{N+1}$, the secrecy capacity minimization can be rewritten as an equality constrained optimization problem, given by
\begin{subequations}
	\label{min_problem}
	\begin{align}
		\text{(SP3)}:\, &\min_{\B{p}_{\rm E},\B{s}\succeq0}\ C_{\rm S} \left(\bar\phi, \B{\bar p}_{\rm F}, \B{p}_{\rm E}\right)\\
		&\ {\rm s.t.}\
		R_{\rm S}^2 - \left\|\B{p}_{\rm E}-\B{p}_{\rm B}\right\|^2 + s_0 = 0,\label{equ_rcpz_const}\\
		&\qquad d_{\rm R}^2 - \left\|\B{p}_{\rm E}-\B{p}_{\rm A}^k\right\|^2 + s_k = 0,\forall k\in\left\{1,\ldots,N\right\}.\label{equ_ff_const}
	\end{align}
\end{subequations}

We now leverage the augmented Lagrangian (AL) method~\cite{ZaiwenWen10} to iteratively solve (SP3) and provide a \textit{Karush-Kuhn-Tucker} (KKT) solution. To begin with, let $\mathcal{L}_\varrho\left( \B{p}_{\rm E},\B{s};\BS\mu\right)$ denote the AL function for (SP3) with positive penalty parameter $\varrho$. Here, the dual variables $\BS\mu = \left[\mu_0,\mu_1,\ldots,\mu_N\right]\T\in\mathbb{R}^{N+1}$ are introduced to handle the equality constraints (\ref{equ_rcpz_const}) and (\ref{equ_ff_const}). The corresponding \textit{AL problem} is defined as
\begin{align}
	\text{(SP4)}: \min_{\B{p}_{\rm E},\B{s}\succeq0}\Big\{&\mathcal{L}_\varrho\left( \B{p}_{\rm E},\B{s};\BS\mu\right) = C_{\rm S} \left(\bar\phi, \B{\bar p}_{\rm F}, \B{p}_{\rm E}\right)\notag\\
	&+ \mu_0\left( R_{\rm S}^2 - \left\|\B{p}_{\rm E}-\B{p}_{\rm B}\right\|^2 + s_0 \right)\notag\\
	&+ \sum_{k=1}^{N}\mu_k\left(d_{\rm R}^2 - \left\|\B{p}_{\rm E}-\B{p}_{\rm A}^k\right\|^2 + s_k\right)\notag\\
	& + \frac{\varrho}{2}\, Q\left(\B{p}_{\rm E}, \B{s}\right)\Big\},
	\label{AL_problem_1}
\end{align}
where $Q\left(\B{p}_{\rm E}, \B{s}\right) = (R_{\rm S}^2 - \left\|\B{p}_{\rm E}-\B{p}_{\rm B}\right\|^2 + s_0)^2 + \sum_{k=1}^{N} (d_{\rm R}^2 - \left\|\B{p}_{\rm E}-\B{p}_{\rm A}^k\right\|^2 + s_k)^2$ is the quadratic penalty function for constraints (\ref{equ_rcpz_const}) and (\ref{equ_ff_const}). The slack variables $\B{s}^{\rm opt}$ satisfy the global optimality conditions of (SP4) if and only if:
\begin{align}
	&s_0^{\rm opt} = \left(\|\B{p}_{\rm E}-\B{p}_{\rm B}\|^2 - R_{\rm S}^2\ - \mu_0/\varrho\right)^+,\\
	&s_k^{\rm opt} = \left(\|\B{p}_{\rm E}-\B{p}_{\rm A}^k\|^2 - d_{\rm R}^2\ - \mu_k/\varrho\right)^+, k\in\left\{1,\ldots,N\right\}.
\end{align}
As a result, the slack variables can be eliminated from the AL function via the substitution $\mathcal{L}_\varrho\left( \B{p}_{\rm E};\BS\mu\right) = \mathcal{L}_\varrho\left( \B{p}_{\rm E},\B{s}^{\rm opt};\BS\mu\right)$, which leads to the simplified AL problem, given by
\begin{align}
	\label{AL_problem_2}
	\text{(SP5)}:\ &\min_{\B{p}_{\rm E}} \bigg\{\mathcal{L}_\varrho\left( \B{p}_{\rm E};\BS\mu\right) = C_{\rm S} \left(\bar\phi, \B{\bar p}_{\rm F}, \B{p}_{\rm E}\right)\notag\\
	&+ \frac{\varrho}{2}\bigg[ \left( R_{\rm S}^2 - \|\B{p}_{\rm E}-\B{p}_{\rm B}\|^2 + \frac{\mu_0}{\varrho} \right)^2 - \left(\frac{\mu_0}{\varrho}\right)^2\notag\\
	&+ \sum_{k=1}^{N} \left( d_{\rm R}^2 - \|\B{p}_{\rm E}-\B{p}_{\rm A}^k\|^2 + \frac{\mu_k}{\varrho} \right)^2 - \left(\frac{\mu_k}{\varrho}\right)^2 \bigg] \bigg\}.
\end{align}

The KKT solution for the secrecy capacity minimization subproblem in (SP2), i.e., the primal dual stationary pair $\left( \B{p}_{\rm E}^\star, \BS\mu^\star \right)$, can be obtained by the AL method through double-loop iterations: The outer loop adjusts the dual variables and the penalty parameter of the AL function, and the inner loop employs gradient descent to iteratively solve the AL problem in (SP5) by updating the primal variables. At the $r$-th iteration of the outer loop, depending on the level of constraint violations measured by
\begin{align}
	V_r\left(\B{p}_{\rm E}\right) = \Bigg(&\max\left\{ R_{\rm S}^2 - \left\|\B{p}_{\rm E}-\B{p}_{\rm B}\right\|^2, -\frac{\mu_0^r}{\varrho_r} \right\}^2\notag\\
	+ \sum_{k=1}^{N}&\max\left\{d_{\rm R}^2 - \|\B{p}_{\rm E}-\B{p}_{\rm A}^k\|^2, -\frac{\mu_k^r}{\varrho_r}\right\}^2\Bigg)^{1/2},
\end{align}
the dual variables either remain unchanged or are updated by
\begin{align}
	&\mu_0^{r+1} = \left[\mu_0^r + \varrho_r\left( R_{\rm S}^2 - \left\|\B{p}_{\rm E}-\B{p}_{\rm B}\right\|^2 \right)\right]^+,\\
	&\mu_k^{r+1} = \left[\mu_k^r + \varrho_r\left( d_{\rm R}^2 - \left\|\B{p}_{\rm E}-\B{p}_{\rm A}^k\right\|^2 \right)\right]^+.
\end{align}
At the $t$-th iteration of the inner loop, the gradient descent step for $\B{p}_{\rm E}$ is
\begin{align}
	\B{p}_{\rm E}^{t+1} = \B{p}_{\rm E}^t - \beta_{\rm E}^t\nabla_{\B{p}_{\rm E}} \mathcal{L}_{\varrho_r}\left( \B{p}_{\rm E}^t; \BS\mu^r \right),
\end{align}
where $\beta_{\rm E}^t$ denotes the \textit{Armijo} step size at the $t$-th iteration chosen by the backtracking line search method, and $\nabla_{\B{p}_{\rm E}} \mathcal{L}_\varrho$ is the gradient of the AL function $\mathcal{L}_\varrho\left(\B{p}_{\rm E};\BS\mu\right)$ with respect to Eve's position, the expression for which is provided in (\ref{grad_AL}) in Appendix~\ref{appendix_C}. The structure of the AL method for solving the secrecy capacity minimization subproblem in (SP2) is summarized in Algorithm~\ref{alg_ALM}.

\begin{algorithm}[t]
	\caption{AL Method for Subproblem (SP2)}
	\begin{algorithmic}[1]
		\State \textbf{initialization:} Initialize $\tilde{\B p}_{\rm E}^0 = \B{p}_{{\rm E},{\rm J}}^l$, $\BS\mu^0$ and $\varrho_0>0$. Choose constraint violation parameter $\eta>0$, precision parameter $\varepsilon>0$ and constants $0<a\leq b\leq1<c$. Set $\eta_0=1/\varrho_0^a$ and $\varepsilon_0=1/\varrho_0$.
		\Repeat\ $r=0,1,2,\ldots,R$\hfill\textit{// AL: gradient descent}
		\State $\B{\check p}_{\rm E}^0 = \tilde{\B p}_{\rm E}^r$
		\Repeat\ $t=0,1,2,\ldots,T$
		\State $\B{\check p}_{\rm E}^{t+1} = \B{\check p}_{\rm E}^t - \beta_{\rm E}^t\nabla_{\B{p}_{\rm E}} \mathcal{L}_{\varrho_r}\left( \B{\check p}_{\rm E}^t; \BS\mu^r \right)$
		\Until $\|\nabla_{\B{p}_{\rm E}} \mathcal{L}_{\varrho_r}\left(\B{\check p}_{\rm E}^{t+1};\BS\mu^r\right)\| \leq \varepsilon_r$
		\State $\tilde{\B p}_{\rm E}^{r+1} = \B{\check p}_{\rm E}^{t+1}$
		\If{$V_r\left(\tilde{\B p}_{\rm E}^{r+1}\right)\leq\eta_r$}
		\Statex \qquad\quad\textit{// Update dual variables:} $k\in\left\{1,\ldots,N\right\}$
		\State $\mu_0^{r+1} = \left[\mu_0^r + \varrho_r\left( R_{\rm S}^2 - \left\|\tilde{\B p}_{\rm E}^{r+1}-\B{p}_{\rm B}\right\|^2 \right)\right]^+$
		\State $\mu_k^{r+1} = \left[\mu_k^r + \varrho_r\left( d_{\rm R}^2 - \left\|\tilde{\B p}_{\rm E}^{r+1}-\B{p}_{\rm A}^k\right\|^2 \right)\right]^+$
		\Statex \qquad\quad\textit{// Update AL parameters}
		\State $\varrho_{r+1} = \varrho_r$, $\eta_{r+1}=\eta_r/\varrho_{r+1}^b$, $\varepsilon_{r+1}=\varepsilon_r/\varrho_{r+1}$
		\Else
		\State $\BS\mu^{r+1} = \BS\mu^r$
		\State $\varrho_{r+1} = c\varrho_r$, $\eta_{r+1}=1/\varrho_{r+1}^a$, $\varepsilon_{r+1}=1/\varrho_{r+1}$
		\EndIf
		\Until $V_r\left(\tilde{\B p}_{\rm E}^{r+1}\right)\leq\eta$ and $\|\nabla_{\B{p}_{\rm E}} \mathcal{L}_{\varrho_r}\left(\tilde{\B p}_{\rm E}^{r+1};\BS\mu^r\right)\| \leq \varepsilon$
		\State \textbf{output:} $\B{p}_{{\rm E},{\rm J}}^{l+1} = \tilde{\B p}_{\rm E}^{r+1}$
	\end{algorithmic}
	\label{alg_ALM}
\end{algorithm}

\subsection{Synchronous GDA for Global Maximin Approximation} \label{section_SGDA}

The GDA framework performs the gradient ascent and gradient descent steps in an alternating way: (i) the projected gradient ascent follows the GP method for the outer maximization as described in Section~\ref{section_GP}, and (ii) the gradient descent follows the AL method for the inner minimization as described in Section~\ref{section_AL}. As mentioned earlier, directly applying the conventional GDA method is more likely to produce convergence to a local and not necessarily the global maximin solution $\left(\phi^\star,\B{p}_{\rm F}^\star,\B{p}_{\rm E}^\star\right)$ for the problem in (P3). As such, there may exist another feasible eavesdropping position $\B{p}_{\rm E}^\prime\neq\B{p}_{\rm E}^\star$ such that $C_{\rm S}\left(\phi^\star,\B{p}_{\rm F}^\star,\B{p}_{\rm E}^\prime\right) < C_{\rm S}\left(\phi^\star,\B{p}_{\rm F}^\star,\B{p}_{\rm E}^\star\right)$. This is due to the fact that only a single candidate eavesdropping position is updated at each iteration in the conventional GDA framework. Ideally, infinitely many potential eavesdropping positions should be taken into account to obtain the global maximin stationary solution $\left(\phi^{\rm max},\B{p}_{\rm F}^{\rm max},\B{p}_{\rm E}^{\rm min}\right)$, which is obviously not possible.

For the sake of analyzing potential eavesdropping positions, we visualize the near-field beamfocusing effect across the $xz$-plane in Fig.~\ref{fig_BFC_gain}, where Alice is located at the origin $\B{p}_{\rm A} = \left(0,0,0\right)$ as shown in Fig. \ref{fig_model}, and Bob is located at $\B{p}_{\rm B} = \left(0,0,10\right)$. Fig.~\ref{fig_beamfocusing_gain} shows the normalized beamfocusing gain when only the information bearing signal is transmitted by Alice with the FP set at the position of Bob, i.e., $\phi=1$ and $\B{p}_{\rm F}=\B{p}_{\rm B}$. Fig.~\ref{fig_beamnulling_gain} demonstrates the normalized \textit{beam nulling} gain when only the null-space AN signal is transmitted by Alice, i.e., $\phi=0$. It can be observed that the behavior of beamfocusing is completely opposite to that of beam nulling: at positions along the ray $\mathcal{R}_{\rm B}$ and in the vicinity of Bob, beamfocusing creates a peak of signal strength for the information bearing signal while beam nulling results in a trough of interference strength for the AN signal. Fortunately, owing to the pencil-like beams generated by the large-scale antenna array, it can be intuitively observed that the worst-case location for Eve is most likely in one of two small regions around the following two positions: $\B{p}_{\rm 1} = \left(1-R_{\rm S}/\left\|\B{p}_{\rm B}\right\|\right)\B{p}_{\rm B}$ and $\B{p}_{\rm 2} = \left(1+R_{\rm S}/\left\|\B{p}_{\rm B}\right\|\right)\B{p}_{\rm B}$, which denote the intersection points between the ray $\mathcal{R}_{\rm B}$ and the RCPZ border, since the SINR decreases sharply as a receiver moves away from those areas. Hence, the most harmful eavesdropping locations would be near the RCPZ and the ray $\mathcal{R}_{\rm B}$, i.e., the regions around $\B{p}_{\rm 1}$ and $\B{p}_{\rm 2}$. For simplicity, we respectively model this pair of regions as Euclidean balls $\mathcal{B}_1$ and $\mathcal{B}_2$ that are centered at $\B{p}_{\rm 1}$ and $\B{p}_{\rm 2}$ with radii $\delta_{\rm 1}$ and $\delta_{\rm 2}$, i.e., $\mathcal{B}_{\rm J} = \left\{ \B{p}\in\mathbb{R}^3| \left\|\B{p}-\B{p}_{\rm J}\right\|\leq\delta_{\rm J}\right\}$, ${\rm J}\in\left\{1,2\right\}$. Recall that the RCPZ is an eavesdropper-free area given by $\mathcal{Z} = \left\{\B{p}\in\mathbb{R}^3| \left\|\B{p}-\B{p}_{\rm B}\right\| < R_{\rm S}\right\}$, the most harmful eavesdropping regions can be denoted as $\mathcal{E}_1 = \mathcal{B}_1\setminus\mathcal{Z}$ and $\mathcal{E}_2 = \mathcal{B}_2\setminus\mathcal{Z}$. Motivated by this, a synchronous GDA (SGDA) framework is proposed to obtain an approximation of the global maximin solution to the problem in (P3) by synchronously updating a pair of randomly initialized eavesdropping positions within the regions $\mathcal{E}_1$ and $\mathcal{E}_2$.

\begin{figure}[t]
	\centering
	\subfigure[Normalized beamfocusing gain of the information bearing signal.]{
		\begin{minipage}[t]{0.48\linewidth}
			\centering
			\includegraphics[width=1\columnwidth]{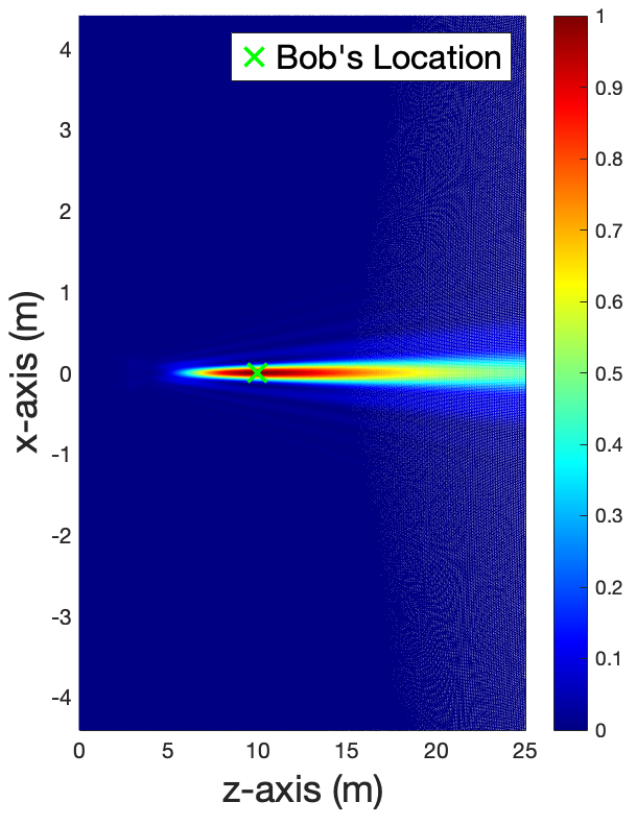}
			\label{fig_beamfocusing_gain}
	\end{minipage}}
	\subfigure[Normalized beam nulling gain of the null-space AN.]{
		\begin{minipage}[t]{0.48\linewidth}
			\centering
			\includegraphics[width=1\columnwidth]{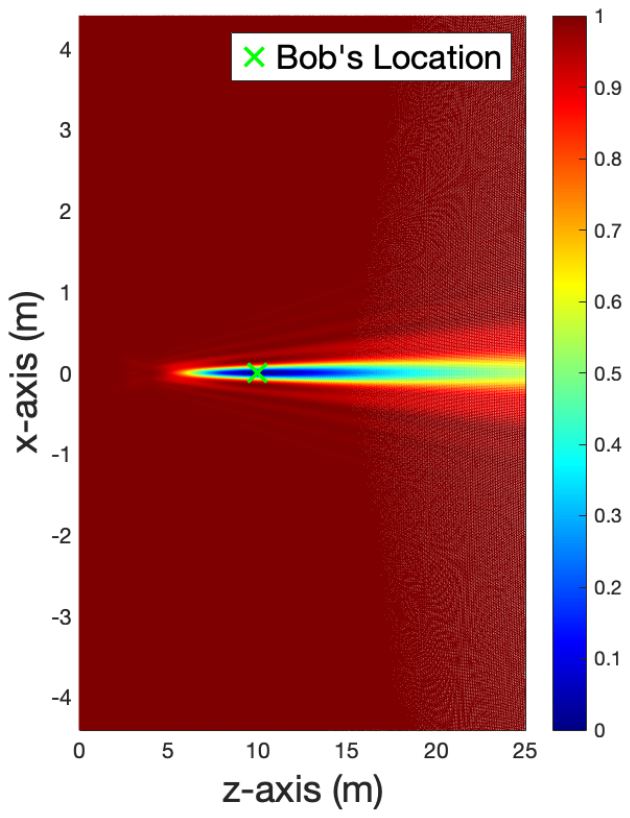}
			\label{fig_beamnulling_gain}
	\end{minipage}}
	\centering
	\caption{The impact of beamfocusing and beam nulling on information signal and AN signal transmissions, respectively.}
	\label{fig_BFC_gain}
\end{figure}

The structure of the proposed SGDA framework is summarized in Algorithm~\ref{alg_SGDA}. It outputs a solution $(\hat\phi^{\rm max},\hat{\B{p}}_{\rm F}^{\rm max},\hat{\B{p}}_{\rm E}^{\rm min})$ that approximates the global optimum, which we refer to as the Maximin solution. For the inner minimization, the resulting solution for Eve's position $\hat{\B{p}}_{\rm E}^{\rm min}$ is guaranteed to be a KKT solution satisfying $0 \approx \|\nabla_{\B{p}_{\rm E}} \mathcal{L}_\varrho( \hat\phi^{\rm max}, \hat{\B{p}}_{\rm F}^{\rm max}, \hat{\B{p}}_{\rm E}^{\rm min}; \BS\mu^\star)\| \leq \varepsilon$. For the outer maximization, the resulting power allocation factor $\hat\phi^{\rm max}$ optimally allocates the transmit power between the information bearing signal and the null-space AN signal, and the resulting FP position $\hat{\B{p}}_{\rm F}^{\rm max}$ strikes a balance between the worst-case secrecy capacity for the eavesdropping positions originating from $\mathcal{E}_1$ and $\mathcal{E}_2$. The computational complexity of the Maximin approach is of order $\mathcal{O}\left(I_{\rm GP}I_{\rm AL}N^2\right)$, where $I_{\rm GP}$ and $I_{\rm AL}$ represent the number of iterations executed by the GP and AL methods, respectively.

\begin{algorithm}[t]
	\caption{SGDA Framework for Problem (P3)}
	\begin{algorithmic}[1]
		\State \textbf{initialization:} Initialize $\B{p}_{\rm F}^0\in\mathcal{R}_{\rm B}$. Randomly sample $\B{p}_{{\rm E},1}^0\in\mathcal{E}_1$ and $\B{p}_{{\rm E},2}^0\in\mathcal{E}_2$ using a uniform distribution. Initialize $\B{p}_{\rm E}^0$ with either $\B{p}_{{\rm E},1}^0$ or $\B{p}_{{\rm E},2}^0$. Choose precision parameter $\varepsilon^\prime>0$.
		\Repeat\ $l=0,1,2,\ldots,L$
		\State $\phi^{l+1} = \Phi\left(\B{p}_{\rm F}^l, \B{p}_{\rm E}^l\right)$ is obtained by the optimal power
		\Statex\quad\, allocation strategy given $\B{p}_{\rm F}^l$, $\B{p}_{\rm E}^l$.
		\Statex\quad\, \textit{// GP: projected gradient ascent}
		\State $\B{p}_{\rm F}^{l+1} = \BS{P}_{\mathcal{R}_{\rm B}}\left[\B{p}_{\rm F}^l + \beta_{\rm F}^l\nabla_{\B{p}_{\rm F}} C_{\rm S}\left(\phi^{l+1},\B{p}_{\rm F}^l,\B{p}_{\rm E}^l\right)\right]$
		\ForAll {${\rm J}\in\left\{1,2\right\}$} \hfill\textit{// Synchronous update}
		\State Solve the subproblem in (SP2) by the AL method
		\Statex\qquad\quad given $\phi^{l+1}$, $\B{p}_{\rm F}^{l+1}$ using $\B{p}_{{\rm E},{\rm J}}^l$ as the initial position \Statex\qquad\quad for Eve.
		\EndFor
		\State $\B{p}_{\rm E}^{l+1} = \arg\min_{\B{p}_{\rm E}\in\left\{\B{p}_{{\rm E},1}^{l+1},\B{p}_{{\rm E},2}^{l+1}\right\}} C_{\rm S}\left(\phi^{l+1},\B{p}_{\rm F}^{l+1},\B{p}_{\rm E}\right)$
		\Until $\left|C_{\rm S}(\phi^{l+1},\B{p}_{\rm F}^{l+1},\B{p}_{{\rm E},1}^{l+1}) - C_{\rm S}(\phi^{l+1},\B{p}_{\rm F}^{l+1},\B{p}_{{\rm E},2}^{l+1})\right| \hspace{-0.04cm} \leq \varepsilon^\prime$
		\State \textbf{output:} $\left(\hat\phi^{\rm max},\hat{\B{p}}_{\rm F}^{\rm max},\hat{\B{p}}_{\rm E}^{\rm min}\right) = \left(\phi^{l+1},\B{p}_{\rm F}^{l+1},\B{p}_{\rm E}^{l+1}\right)$
	\end{algorithmic}
	\label{alg_SGDA}
\end{algorithm}

\begin{remark}
	For the special case where no AN signal is transmitted by Alice, the proposed SGDA framework remains valid when the power allocation factor is set as $\phi=1$ in Algorithm~\ref{alg_SGDA}. In this case, the corresponding computational complexity of the Maximin solution can be further reduced to $\mathcal{O}\left(I_{\rm GP}I_{\rm AL}N\right)$.
\end{remark}

\section{Low-Complexity Design} \label{section_RCPZ_low_complexity}

In addition to the approximate optimal solution elaborated in the previous section, we provide a low-complexity solution in this section.

\subsection{Low-Complexity Beamfocusing Design} \label{section_low_complexity_BFC}

Motivated by the discussion in Section~\ref{section_SGDA} on the most harmful eavesdropping regions, we can simplify the max-min problem in (P3) by restricting the eavesdropper to be located exactly at either of the intersection points between the ray $\mathcal{R}_{\rm B}$ and the RCPZ border, leading to
\begin{subequations}
	\label{max_min_fp_low_complexity_problem}
	\begin{align}
		\text{(P4)}:\ \max_{\phi, \B{p}_{\rm F}}\ &\min_{\B{p}_{\rm E}}\ C_{\rm S} \left(\phi, \B{p}_{\rm F}, \B{p}_{\rm E}\right)\\
		{\rm s.t.}\
		& \B{p}_{\rm E} \in \left\{\B{p}_1,\B{p}_2\right\},\label{intersection_const}\\
		& {\rm \left(\ref{alloc_const}\right)}, {\rm \left(\ref{fp_const}\right)}, \notag
	\end{align}
\end{subequations}
where the constraints (\ref{rcpz_const}) and (\ref{ff_const}) of the problem in (P3) are replaced with the constraint (\ref{intersection_const}). Then, by enumerating the finite feasible set specified by (\ref{intersection_const}), the max-min problem in (P4) can be equivalently written as the following maximization problem
\begin{subequations}
	\label{max_fp_low_complexity_problem}
	\begin{align}
		\text{(P5)}:\ &\max_{\phi, \B{p}_{\rm F}}\ \tilde{C}_{\rm S}\left(\phi, \B{p}_{\rm F}\right)\\
		&\ {\rm s.t.}\
		{\rm \left(\ref{alloc_const}\right)}, {\rm \left(\ref{fp_const}\right)}, \notag
	\end{align}
\end{subequations}
where the objective function is defined by $\tilde{C}_{\rm S}\left(\phi, \B{p}_{\rm F}\right) = \min\left\{C_{\rm S} \left(\phi, \B{p}_{\rm F}, \B{p}_1\right), C_{\rm S} \left(\phi, \B{p}_{\rm F}, \B{p}_2\right)\right\}$. Instead of optimizing the power allocation factor and the FP in an alternating way, which is commonly-adopted in the literature (e.g.,~\cite{XianghaoYu16}) but tends to incur high computational complexity, a divide-and-conquer strategy is applied to approximate the solution to (P5) in a one-shot manner. To begin with, an upper bound on $\tilde{C}_{\rm S}\left(\phi, \B{p}_{\rm F}\right)$ is provided in the following proposition.

\begin{proposition}
	For any feasible power allocation factor and FP satisfying the constraints (\ref{alloc_const}) and (\ref{fp_const}), the objective function of (P5) is upper bounded by $\tilde{C}_{\rm S}^{\rm ub}\left(\B{p}_{\rm F}\right) = \min\left\{C_{\rm S} \left(\phi_1, \B{p}_{\rm F}, \B{p}_1\right), C_{\rm S} \left(\phi_2, \B{p}_{\rm F}, \B{p}_2\right)\right\}$:
	\begin{align}
		\tilde{C}_{\rm S}\left(\phi, \B{p}_{\rm F}\right) \leq \tilde{C}_{\rm S}^{\rm ub}\left(\B{p}_{\rm F}\right),\ \forall \phi\in\left[0,1\right], \B{p}_{\rm F}\in\mathcal{R}_{\rm B},
	\end{align}
	where $\phi_1 = \Phi\left(\B{p}_{\rm F},\B{p}_1\right)$ and $\phi_2 = \Phi\left(\B{p}_{\rm F},\B{p}_2\right)$ are the optimal power allocation factors derived from Proposition~\ref{proposition_2}.
	\label{proposition_3}
\end{proposition}

\begin{proof}
	From Proposition~\ref{proposition_2}, we obtain inequalities for any power allocation factor $\phi\in\left[0,1\right]$ and FP $\B{p}_{\rm F}\in\mathcal{R}_{\rm B}$:
	\begin{align}
		C_{\rm S} \left(\phi, \B{p}_{\rm F}, \B{p}_1\right) \leq C_{\rm S} \left(\Phi\left(\B{p}_{\rm F},\B{p}_1\right), \B{p}_{\rm F}, \B{p}_1\right),\\
		C_{\rm S} \left(\phi, \B{p}_{\rm F}, \B{p}_2\right) \leq C_{\rm S} \left(\Phi\left(\B{p}_{\rm F},\B{p}_2\right), \B{p}_{\rm F}, \B{p}_2\right).
	\end{align}
	Therefore, $\tilde{C}_{\rm S}\left(\phi, \B{p}_{\rm F}\right) \leq \tilde{C}_{\rm S}^{\rm ub}\left(\B{p}_{\rm F}\right)$ is obtained based on their definitions, which completes the proof.
\end{proof}

The proposed low-complexity design consists of two sequential one-dimensional searches for: (i) the FP $\B{p}_{\rm F}^\star = \argmax_{\B{p}_{\rm F}\in\mathcal{R}_{\rm B}} \tilde{C}_{\rm S}^{\rm ub}\left(\B{p}_{\rm F}\right)$, and (ii) the power allocation factor $\phi^\star = \argmax_{0\leq\phi\leq1} \tilde{C}_{\rm S}\left(\phi, \B{p}_{\rm F}^\star\right)$ given $\B{p}_{\rm F}^\star$. The rationale behind this solution is that step (i) pursues the FP over the upper bound $\tilde{C}_{\rm S}^{\rm ub}$ by optimistically assuming that the power allocation strategy is sufficiently good, after which step (ii) further improves the power allocation strategy based on the resulting FP.

In practice, the proposed low-complexity solution always results in identical secrecy capacities $C_{\rm S} \left(\phi^\star, \B{p}_{\rm F}^\star, \B{p}_1\right) = C_{\rm S} \left(\phi^\star, \B{p}_{\rm F}^\star, \B{p}_2\right)$, which indicates that the received SINRs at both $\B{p}_1$ and $\B{p}_2$ are the same. Achieving the same SINR at $\B{p}_1$ and $\B{p}_2$ constrains the highest possible SINR at Eve while providing a relatively strong SINR at Bob. We refer to this low-complexity solution as Equal-SINRs.

\begin{remark}
	For the special case where no AN signal is transmitted by Alice, the equation $\tilde{C}_{\rm S}\left(\phi, \B{p}_{\rm F}\right) = \tilde{C}_{\rm S}^{\rm ub}\left(\B{p}_{\rm F}\right),\ \forall \B{p}_{\rm F}\in\mathcal{R}_{\rm B}$ holds as a consequence of $\phi=\phi_1=\phi_2=1$. In this case, the proposed Equal-SINRs solution reduces to achieving the same received SNR at both $\B{p}_1$ and $\B{p}_2$, which we refer to as the Equal-SNRs solution.
\end{remark}

\subsection{Computational Complexity Analysis}

The proposed low-complexity Equal-SINRs and Equal-SNRs designs require computational complexity of order $\mathcal{O}\left(N^2+C_{\rm F}N\right)$ and $\mathcal{O}\left(C_{\rm F}N\right)$, respectively, where $C_{\rm F}$ is a constant determined by the FP search precision. Note that only the secrecy capacity minimization subproblem in (SP2) needs to be solved for the proposed low-complexity approaches, since the beamfocusing vector and the power allocation factor are determined directly. The computational complexity of the various approaches discussed in the paper are listed in Table~\ref{table_complexity_RCPZ}, illustrating the significant savings of the Equal-SINRs and Equal-SNRs approaches.

\begin{table}[t]
	\centering
	\caption{Computational Complexity of Proposed Solutions with RCPZ}
	\begin{tabular}{|l|l|c|}
		\hline
		\textbf{Solution with RCPZ} & \textbf{Complexity Order} & \textbf{Null-Space AN}\\
		\hline
		\text{Maximin-w/-AN} & $\mathcal{O}\left(I_{\rm GP}I_{\rm AL}N^2\right)$ & $\checkmark$\\
		\hline
		\text{Maximin-w/o-AN} & $\mathcal{O}\left(I_{\rm GP}I_{\rm AL}N\right)$ & $\times$\\
		\hline
		\text{Equal-SINRs} & $\mathcal{O}\left(N^2+C_{\rm F}N\right)$ & $\checkmark$\\
		\hline
		\text{Equal-SNRs} & $\mathcal{O}\left(C_{\rm F}N\right)$ & $\times$\\
		\hline
		\text{Peak-SNR} & $\mathcal{O}\left(C_{\rm F}N\right)$ & $\times$\\
		\hline
		\text{Conventional-MRT} & $\mathcal{O}\left(1\right)$ & $\times$\\
		\hline
	\end{tabular}
	\label{table_complexity_RCPZ}
\end{table}

\section{Full Duplex Enabled Receiver-Centered Virtual Protected Zone} \label{section_RCVPZ}

In previous sections, we have studied how to design secure transmission when the system has an RCPZ in place to physically prevent Eve from approaching too close to Bob. In this section, we extend the study to consider a more challenging scenario where a physical RCPZ is no longer available. To guarantee security in this situation, we consider Bob to possess a full-duplex transceiver capable of simultaneously receiving information and transmitting AN \cite{ShihaoYan18, GanZheng13, WeiLi12}. The AN generated from Bob effectively creates a \textit{Receiver-Centered Virtual Protected Zone} (RCVPZ) that discourages Eve from locating herself near Bob. In the subsequent discussion, we define the \textit{Virtual Security Radius} of an RCVPZ as the distance from Bob to the worst-case eavesdropping location.

\subsection{System Model and Problem Formulation with RCVPZ}

We reconsider the near-field downlink communication system in Section~\ref{section_system_model}, where now a full-duplex Bob is able to simultaneously receive information and transmit AN omnidirectionally through a single antenna. Such a shared-antenna full duplex structure can be realized using circulators~\cite{Bharadia13, Sabharwal14}, and has been considered for practical implementation in mobile devices of limited size~\cite{Korpi16}.

The signals received by Bob and Eve are given by
\begin{align}
	\check{r}_{\rm B} = \B{h}_{\rm B}\CT \B{w}s + h_{\rm BB}u + n_{\rm B}
\end{align}
and
\begin{align}
	\check{r}_{\rm E} = \B{h}_{\rm E}\CT \B{w}s + \B{h}_{\rm E}\CT\B{z} + h_{\rm BE}u + n_{\rm E},
\end{align}
respectively, where $u\sim\mathcal{CN}\left(0,P_{\rm B}\right)$ is the AN signal of power $P_{\rm B}$ transmitted by Bob. The effective self-interference channel is modeled by $h_{\rm BB}$ with amplitude $|h_{\rm BB}|=\sqrt{\rho}$, and $0<\rho\leq1$ parameterizes the effectiveness of SIC at Bob~\cite{ShihaoYan18, GanZheng13}, which is expressed in dB as $10\log_{10}\frac{1}{\rho}$. Thus, the residual self-interference is expressed as $h_{\rm BB}u \sim\mathcal{CN}\left(0,\rho P_{\rm B}\right)$. The LoS channel between Bob and Eve is modeled by $h_{\rm BE} = \alpha_{\rm BE}\exp\left(-j\kappa d_{\rm BE}\right)$ with $\alpha_{\rm BE} = (2\kappa d_{\rm BE})^{-1}$ and $d_{\rm BE} = \|\B{p}_{\rm B}-\B{p}_{\rm E}\|$ representing the free-space path gain and the distance between Bob and Eve, respectively.

The secrecy capacity for transmission from Alice to Bob can be written as $\check{C}_{\rm S}=(\check{C}_{\rm B}-\check{C}_{\rm E})^+$, where
\begin{align}
	\check{C}_{\rm B} = \log_2\left(1 + \frac{|\B{h}_{\rm B}\CT \B{w}_{\rm F}|^2}{\rho P_{\rm B} + \sigma_{\rm B}^2}\right)
\end{align}
and
\begin{align}
	\check{C}_{\rm E} = \log_2\left(1 + \frac{|\B{h}_{\rm E}\CT \B{w}_{\rm F}|^2}{\frac{\left(1-\phi\right)P_{\rm A}}{N-1} \|\B{h}_{\rm E}\CT\B{Z}\|^2 + \alpha_{\rm BE}^2 P_{\rm B} + \sigma_{\rm E}^2}\right)
\end{align}
are individual channel capacities for Bob and Eve, respectively. Then, the FP-based beamfocusing design for the worst-case secrecy capacity can be formulated as
\begin{subequations}
	\label{max_min_fp_fd_problem}
	\begin{align}
		\text{(P6)}:\ \max_{\phi, \B{p}_{\rm F}}\ &\min_{\B{p}_{\rm E}}\ \check{C}_{\rm S} \left(\phi, \B{p}_{\rm F}, \B{p}_{\rm E}\right)\\
		{\rm s.t.}\
		& {\rm \left(\ref{alloc_const}\right)}, {\rm \left(\ref{ff_const}\right)}, {\rm \left(\ref{fp_const}\right)}. \notag
	\end{align}
\end{subequations}
Comparing the problem in (P6) with that in (P3), we note that constraint (\ref{rcpz_const}) is removed from (P6) since the physical dimension of the RCVPZ is not specified.

\subsection{Extended SGDA-Based Beamfocusing Design with RCVPZ} \label{section_extended_SGDA}

In this section, the SGDA framework elaborated in Section~\ref{section_RCPZ_SGDA} is extended to beamfocusing design with an RCVPZ. To begin, we note that the optimal power allocation strategy offered by Proposition~\ref{proposition_2} is still applicable to the scenario considered here by replacing $\sigma_{\rm B}^2$ and $\sigma_{\rm E}^2$ with $\check\sigma_{\rm B}^2 = \sigma_{\rm B}^2 + \rho P_{\rm B}$ and $\check\sigma_{\rm E}^2 = \sigma_{\rm E}^2 + \alpha_{\rm BE}^2 P_{\rm B}$, respectively.

Due to the lack of a physically imposed RCPZ, the most harmful eavesdropping locations are not necessarily the same as in the discussion of Section~\ref{section_SGDA}, so (P6) is not well suited for the SGDA framework proposed in Section~\ref{section_RCPZ_SGDA}. Extra effort is required to identify the most harmful eavesdropping regions so that the SGDA framework can be applied to (P6). Consequently, we simplify the max-min problem in (P6) by restricting the eavesdropper to lie along the ray $\mathcal{R}_{\rm B}$, leading to
\begin{subequations}
	\label{max_min_fp_fd_heuristic_problem}
	\begin{align}
		\text{(P7)}:\ \max_{\phi, \B{p}_{\rm F}}\ &\min_{\B{p}_{\rm E}}\ \check{C}_{\rm S} \left(\phi, \B{p}_{\rm F}, \B{p}_{\rm E}\right)\\
		{\rm s.t.}\
		& \B{p}_{\rm E} \in \mathcal{R}_{\rm B}, \label{eve_rb_const}\\
		& {\rm \left(\ref{alloc_const}\right)}, {\rm \left(\ref{ff_const}\right)}, {\rm \left(\ref{fp_const}\right)}. \notag
	\end{align}
\end{subequations}
We note that Problem (P7) is similar to (P4), but the constraint (\ref{intersection_const}) is replaced with (\ref{eve_rb_const}) since neither of the intersection points can be easily determined.

Following the low-complexity solution elaborated in Section~\ref{section_low_complexity_BFC}, we aim to find the worst-case eavesdropping locations along the ray $\mathcal{R}_{\rm B}$ by solving (P7). The proposed design iteratively searches for the FP and the power allocation factor as in the Equal-SINRs solution until a pair of worst-case eavesdropping positions $\{\B{p}_{{\rm E},1}^\star, \B{p}_{{\rm E},2}^\star\} \in \mathcal{R}_{\rm B}$ on either side of Bob are determined. Thereafter, the most harmful eavesdropping regions can be modeled as Euclidean balls $\mathcal{E}_1$ and $\mathcal{E}_2$ that are centered at $\B{p}_{{\rm E},1}^\star$ and $\B{p}_{{\rm E},2}^\star$ with radii $\delta_{\rm 1}$ and $\delta_{\rm 2}$, i.e., $\mathcal{E}_{\rm J} = \left\{ \B{p}\in\mathbb{R}^3| \left\|\B{p}-\B{p}_{{\rm E},{\rm J}}^\star\right\|\leq\delta_{\rm J}\right\}$, ${\rm J}\in\left\{1,2\right\}$. Given the most harmful eavesdropping regions, the SGDA framework can find the maximin stationary solution $(\hat\phi^{\rm max},\hat{\B{p}}_{\rm F}^{\rm max},\hat{\B{p}}_{\rm E}^{\rm min})$ to the problem in (P6).

The extended SGDA framework for beamfocusing design with RCVPZ described above is summarized in Algorithm~\ref{alg_SGDA_RCVPZ}, where 
\begin{align}
	&\tilde{C}_{{\rm S},l} \left(\phi, \B{p}_{\rm F}\right) = \min\left\{C_{\rm S} \left(\phi, \B{p}_{\rm F}, \B{p}_{{\rm E},1}^l\right), C_{\rm S} \left(\phi, \B{p}_{\rm F}, \B{p}_{{\rm E},2}^l\right)\right\},\\
	&\tilde{C}_{{\rm S},l}^{\rm ub}\left(\B{p}_{\rm F}\right) = \min\left\{C_{\rm S} \left(\phi_1^l, \B{p}_{\rm F}, \B{p}_{{\rm E},1}^l\right), C_{\rm S} \left(\phi_2^l, \B{p}_{\rm F}, \B{p}_{{\rm E},2}^l\right)\right\},
\end{align}
with $\phi_{\rm J}^l = \Phi\left(\B{p}_{\rm F},\B{p}_{{\rm E},{\rm J}}^l\right)$, ${\rm J}\in\left\{1,2\right\}$. The computational complexity of the extended SGDA-based design is of order $\mathcal{O}\left(\left(I_{\rm GP}I_{\rm AL} + C_{\rm E}I_{\rm S}\right)N^2\right)$, where $C_{\rm E}$ is a constant determined by the search precision on the eavesdropping location, and $I_{\rm S}$ stands for the number of iterations executed in searching for the most harmful eavesdropping regions. Expressions for the gradients $\nabla_{\B{p}_{\rm F}} \check{C}_{\rm S}$ and $\nabla_{\B{p}_{\rm E}} \mathcal{\check L}_\varrho$ are provided in (\ref{grad_Cs_pF}) and (\ref{grad_AL}) in Appendix~\ref{appendix_C}, respectively.

\begin{algorithm}[t]
	\caption{Extended SGDA Framework for Problem (P6)}
	\begin{algorithmic}[1]
		\State \textbf{initialization:} Initialize $\B{p}_{{\rm E},1}^0$, $\B{p}_{{\rm E},2}^0 \in \mathcal{R}_{\rm B}$ with $\|\B{p}_{{\rm E},1}^0\| < \left\|\B{p}_{\rm B}\right\|$ and $\|\B{p}_{{\rm E},2}^0\| > \left\|\B{p}_{\rm B}\right\|$. Choose precision parameter $\varepsilon^{\prime\prime} > 0$.
		\Repeat\ $l=0,1,2,\ldots,L$
		\State \parbox[t]{220pt}{Solve $\B{p}_{\rm F}^{l+1} = \argmax_{\B{p}_{\rm F}} \tilde{C}_{{\rm S},l}^{\rm ub} \left(\B{p}_{\rm F}\right)$ for $\B{p}_{\rm F}\in\mathcal{R}_{\rm B}$.\strut}
		\State \parbox[t]{220pt}{Solve $\phi^{l+1} = \argmax_\phi \tilde{C}_{{\rm S},l} \left(\phi, \B{p}_{\rm F}^{l+1}\right)$ for $\phi\in\left[0,1\right]$.\strut}
		\ForAll {${\rm J}\in\left\{1,2\right\}$}
		\State \parbox[t]{205pt}{Solve $\B{p}_{{\rm E},{\rm J}}^{l+1} = \argmax_{\B{p}_{{\rm E},{\rm J}}} C_{\rm S} \left(\phi^{l+1}, \B{p}_{\rm F}^{l+1}, \B{p}_{\rm E}\right)$ for $\B{p}_{{\rm E},{\rm J}}^{l+1} \in \mathcal{R}_{\rm B}$ where $\|\B{p}_{{\rm E},1}\| < \left\|\B{p}_{\rm B}\right\|$ and $\|\B{p}_{{\rm E},2}\| > \left\|\B{p}_{\rm B}\right\|$.\strut}
		\EndFor
		\Until $\|\B{p}_{{\rm E},1}^{l+1} - \B{p}_{{\rm E},1}^l \| \leq \varepsilon^{\prime\prime}$ and $\|\B{p}_{{\rm E},2}^{l+1} - \B{p}_{{\rm E},2}^l \| \leq \varepsilon^{\prime\prime}$
		\State $\B{p}_{\rm F}^\star = \B{p}_{\rm F}^{l+1}$, $\B{p}_{{\rm E},1}^\star = \B{p}_{{\rm E},1}^{l+1}$, $\B{p}_{{\rm E},2}^\star = \B{p}_{{\rm E},2}^{l+1}$
		\State Apply Algorithm~\ref{alg_SGDA} with initial FP $\B{p}_{\rm F}^\star$ and worst-case eavesdropping regions $\mathcal{E}_{\rm J} = \{ \B{p}\in\mathbb{R}^3| \|\B{p}-\B{p}_{{\rm E},{\rm J}}^\star\|\leq\delta_{\rm J}\}$, ${\rm J}\in\left\{1,2\right\}$.
		\State \textbf{output:} $\left(\hat\phi^{\rm max},\hat{\B{p}}_{\rm F}^{\rm max},\hat{\B{p}}_{\rm E}^{\rm min}\right)$
	\end{algorithmic}
	\label{alg_SGDA_RCVPZ}
\end{algorithm}

\begin{remark}
	For the special case where no AN signal is transmitted by Alice, the extended SGDA-based solution remains valid when the power allocation factor is set as $\phi=1$ in Algorithm~\ref{alg_SGDA_RCVPZ}. In this case, the corresponding computational complexity is reduced to $\mathcal{O}\left(\left(I_{\rm GP}I_{\rm AL} + C_{\rm F}I_{\rm S} + C_{\rm E}I_{\rm S}\right)N\right)$.
\end{remark}

\begin{table}[t]
	\centering
	\caption{Simulation Parameters}
	\begin{tabular}{|l|l|}
		\hline
		\textbf{Parameter} & \textbf{Value}\\
		\hline
		\text{Carrier frequency} & $f_{\rm c} = 28\, {\rm GHz}$\\
		\hline
		\text{Number of antennas} & $N_{\rm x} = N_{\rm y} = 128$, $N = 16384$\\
		\hline
		\text{Normalized antenna separation} & $\Delta_{\rm x} = \Delta_{\rm y} = 1/2$\\
		\hline
		\text{Transmit power} & $P_{\rm TX} = 5\, {\rm dBm}$\\
		\hline
		\text{Noise power} & $\sigma_{\rm B}^2 = \sigma_{\rm E}^2 = -75\, {\rm dBm}$\\
		\hline
	\end{tabular}
	\label{tab_param}
\end{table}

\section{Numerical Results} \label{section_numerical_results}

In this section, we first compare the performance of our proposed SGDA-based algorithm with the proposed low-complexity solution for beamfocusing design with RCPZ. Then, we investigate the secrecy performance of beamfocusing design with RCVPZ. The simulation setup is provided in Table~\ref{tab_param}. The worst-case secrecy performance of all solutions is obtained by finding the lowest secrecy capacity for 100 random initializations of Eve's position $\left(\B{p}_{{\rm E},1}^0,\B{p}_{{\rm E},2}^0\right) \in \mathcal{E}_1 \times \mathcal{E}_2$, where the neighborhoods $\mathcal{B}_{\rm 1}$ and $\mathcal{B}_{\rm 2}$ have a radius of $\delta_{\rm 1} = \delta_{\rm 2} = 1{\rm\, m}$.

\begin{figure}[t]
	\centering
	\subfigure[Bob located on the $z$-axis.]{
		\begin{minipage}[t]{0.48\linewidth}
			\centering
			\includegraphics[width=1\columnwidth]{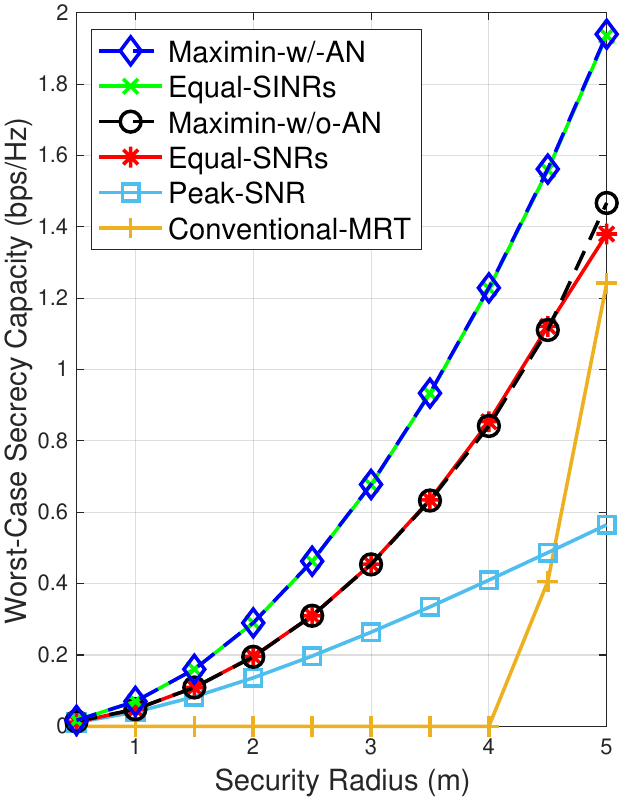}
			\label{fig_radius_a}
	\end{minipage}}
	\subfigure[Bob not located on the $z$-axis.]{
		\begin{minipage}[t]{0.48\linewidth}
			\centering
			\includegraphics[width=1\columnwidth]{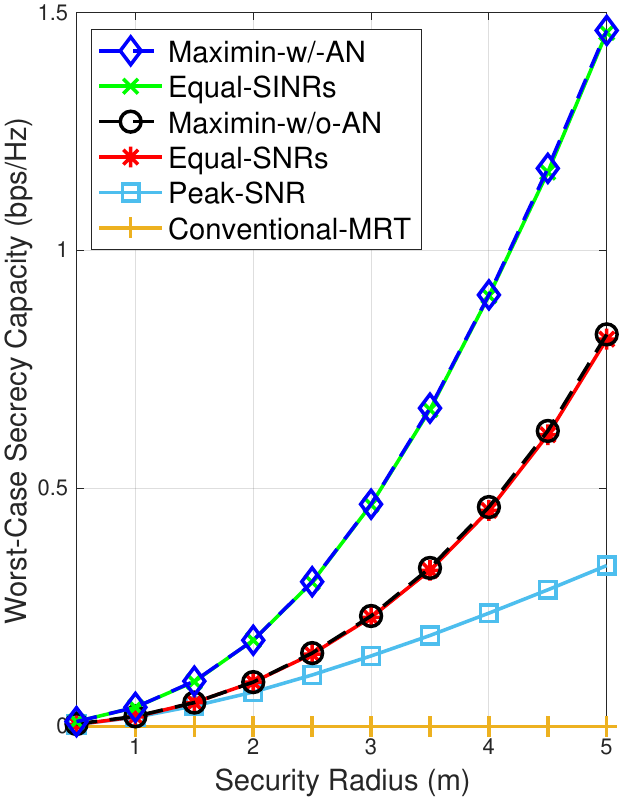}
			\label{fig_radius_b}
	\end{minipage}}
	\centering
	\caption{Worst-case secrecy capacity versus the security radius: a) Bob is located at $\B{p}_{\rm B} = \left(0,0,10\right)$; b) Bob is located at $\B{p}_{\rm B} = (5\sqrt{2}/2,5\sqrt{6}/2,5\sqrt{2})$ with the spherical coordinate $\left(\rho_{\rm B},\theta_{\rm B},\phi_{\rm B}\right) = \left(10,60^\circ,45^\circ\right)$.}
	\label{fig_radius}
\end{figure}

\subsection{Worst-Case Secrecy Performance with RCPZ}

Here we introduce two non-AN-aided (i.e., $\phi=1$) baseline beamfocusing designs for comparison: (i) \textbf{Conventional-MRT}: the most intuitive approach is to set the FP exactly at Bob's position, i.e., $\B{p}_{\rm F} = \B{p}_{\rm B}$, as discussed in the near-field communications literature (e.g.,~\cite{Anaya22, ZhengZhang24, YunpuZhang25, ZhifengTang25}); (ii) \textbf{Peak-SNR}: a similar solution that also aims to solely benefit Bob is to set the FP such that the maximum received SNR is at Bob's position. The computational complexity of both solutions are also listed in Table~\ref{table_complexity_RCPZ}.

In Fig.~\ref{fig_radius}, we compare the worst-case secrecy capacity of all solutions versus the security radius of the RCPZ. Fig.~\ref{fig_radius} investigates the secrecy performance for two distinct positions for Bob, both of which are $10\,{\rm m}$ away from Alice but in different locations. First, we see that the Maximin solutions provide positive secrecy capacities even with a very small RCPZ (e.g., with a security radius less than $1\,{\rm m}$) and their performance improves substantially with a larger RCPZ. This highlights the benefit of carefully-designed near-field beamfocusing for location-based secure communications. In addition, the secrecy capacity offered by the Maximin-w/-AN solution is significantly greater than that of Maximin-w/o-AN, showing the notable benefit of optimally allocating a portion of transmit power for AN to improve secrecy. Furthermore, it is interesting to observe that the secrecy performance is unsatisfactory when the FP is set to Bob's location, i.e., the Conventional-MRT solution, which provides a positive secrecy capacity only if the security radius is relatively large, e.g., $R_{\rm S}>4\,{\rm m}$ in Fig.~\ref{fig_radius_a}. This is due to the fact that, although the choice of the FP means that Bob benefits from maximum antenna array gain~\cite{Bjornson21}, the large-scale path loss causes the maximum SNR point to deviate from the FP, and this deviation can be exploited by Eve when the RCPZ is not large enough to safeguard this position. In contrast, the Peak-SNR solution can guarantee positive secrecy capacity with a small security radius, as demonstrated in both Figs.~\ref{fig_radius_a}~and~\ref{fig_radius_b}. However, Peak-SNR does not take into account the possible positions of Eve and hence offers degraded secrecy performance. We see that the Equal-SINRs and Equal-SNRs solutions both provide much better secrecy performance, and their advantages over Peak-SNR increase with the security radius. We also see that the performance of Equal-SINRs and Equal-SNRs is very close to that of Maximin-w/-AN and Maximin-w/o-AN, highlighting the usefulness of our proposed low-complexity solutions. 

\begin{figure}[t]
	\centering
	\subfigure[Bob located on the $z$-axis.]{
		\begin{minipage}[t]{0.48\linewidth}
			\centering
			\includegraphics[width=1\columnwidth]{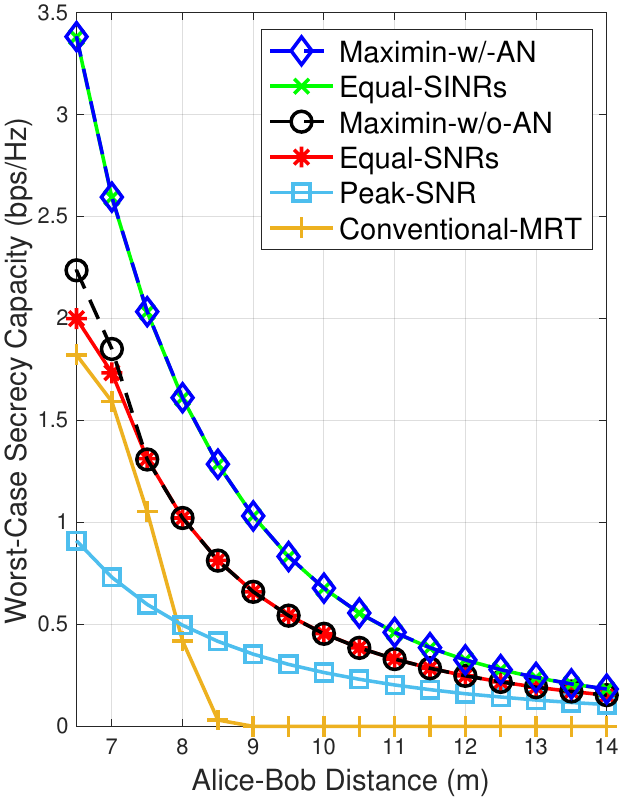}
			\label{fig_dist_a}
	\end{minipage}}
	\subfigure[Bob not located on the $z$-axis.]{
		\begin{minipage}[t]{0.48\linewidth}
			\centering
			\includegraphics[width=1\columnwidth]{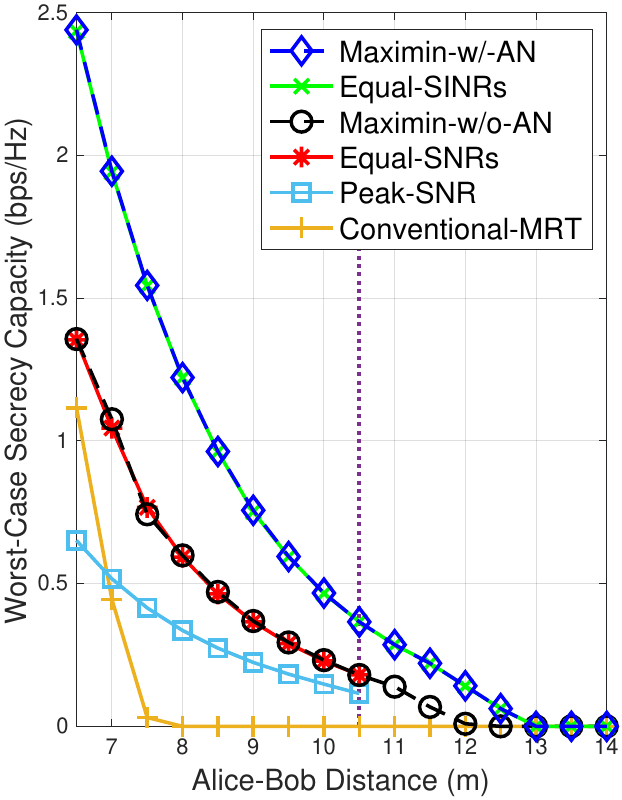}
			\label{fig_dist_b}
	\end{minipage}}
	\centering
	\caption{Worst-case secrecy capacity versus the distance between Alice and Bob when $R_{\rm S} = 3\,{\rm m}$: a) Bob is located at $\B{p}_{\rm B} = \left(0,0,d_{\rm AB}\right)$; b) Bob is located at the spherical coordinate $\left(\rho_{\rm B},\theta_{\rm B},\phi_{\rm B}\right) = \left(d_{\rm AB},60^\circ,45^\circ\right)$.}
	\label{fig_dist}
\end{figure}

\begin{figure}[b]
	\centering
	\subfigure[Minimum secrecy capacity versus the security radius.]{
		\begin{minipage}[t]{0.48\linewidth}
			\centering
			\includegraphics[width=1\columnwidth]{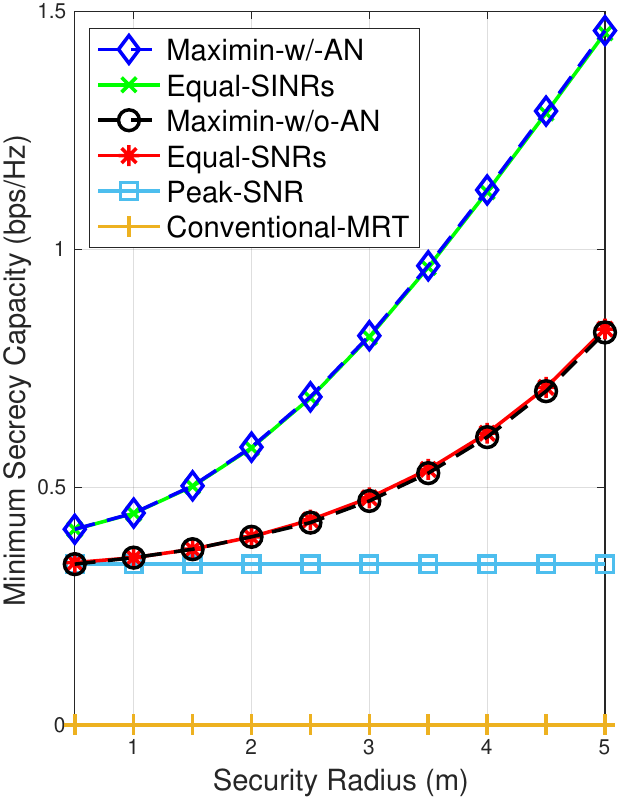}
			\label{fig_multiEve_a}
	\end{minipage}}
	\subfigure[Minimum secrecy capacity versus the distance between Alice and Bob.]{
		\begin{minipage}[t]{0.48\linewidth}
			\centering
			\includegraphics[width=1\columnwidth]{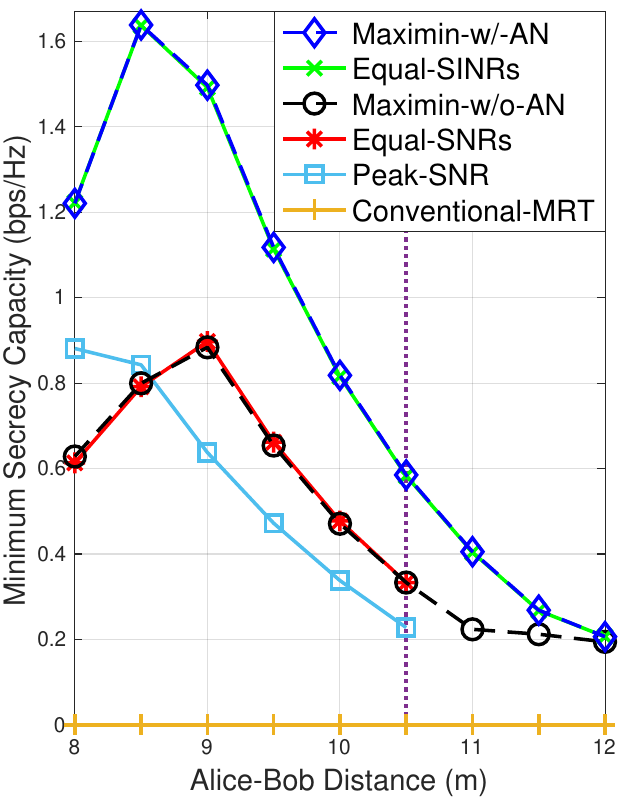}
			\label{fig_multiEve_b}
	\end{minipage}}
	\centering
	\caption{Minimum secrecy capacity with multiple non-colluding eavesdroppers: a) Bob is located at $\B{p}_{\rm B} = (5\sqrt{2}/2,5\sqrt{6}/2,5\sqrt{2})$ with the spherical coordinate $\left(\rho_{\rm B},\theta_{\rm B},\phi_{\rm B}\right) = \left(10,60^\circ,45^\circ\right)$; b) Bob is located at the spherical coordinate $\left(\rho_{\rm B},\theta_{\rm B},\phi_{\rm B}\right) = \left(d_{\rm AB},60^\circ,45^\circ\right)$ with $R_{\rm S} = 3\,{\rm m}$.}
	\label{fig_multiEve}
\end{figure}

In Fig.~\ref{fig_dist}, we compare the worst-case secrecy capacity of all solutions versus the distance between Alice and Bob. Fig.~\ref{fig_dist} explores the secrecy performance for two distinct locations for Bob with the same fixed security radius. Our observations for this case are similar to those in Fig.~\ref{fig_radius}, i.e., it is generally not good to locate either the FP or the maximum SNR point at Bob, while the Equal-SINRs and Equal-SNRs solutions provide good secrecy performance, with only a $0.1\%$-$10.6\%$ reduction in secrecy capacity compared with Maximin-w/-AN and Maximin-w/o-AN. Thus, the low-complexity Equal-SINRs and Equal-SNRs solutions are good candidates for practical implementations. Further discussion and analysis on the approximation of global optimality of Maximin-w/-AN and Maximin-w/o-AN solutions are provided in Appendix~\ref{appendix_D}.

Recall that this work adopts an analog as opposed to a digital beamfocusing design. In order to illustrate the performance difference between these two approaches, we regenerated numerical results for all considered beamfocusing solutions with an FP-based digital beamfocusing vector $\B{w}_{\rm F}^{\rm D} = \sqrt{\phi P_{\rm A}} \frac{\B{h}_{\rm F}}{\|\B{h}_{\rm F}\|}$ under the same settings specified in Figs.~\ref{fig_radius}~and~\ref{fig_dist}. These results are not presented here for brevity. We found that the average performance loss of analog versus digital beamfocusing is less than $1.3\%$ for all cases. The reason for such insignificant losses is that the amplitude/power variations across Alice's antenna elements perceived at the FP position become negligible when the distance from Alice to the FP is comparable to/larger than the \textit{uniform-power distance}~\cite{HaiquanLu22}, which makes analog beamfocusing with uniform element gains a good approximation to its digital counterpart.

In Fig.~\ref{fig_multiEve}, we consider a specific multi-eavesdropper scenario with four non-colluding Eves respectively located at $\frac{1}{2}\B{p}_0$, $\frac{3}{2}\B{p}_0$, $\B{p}_0+\left(5,0,0\right)$ and $\B{p}_0-\left(5,0,0\right)$ outside the RCPZ, where $\B{p}_0 = (5\sqrt{2}/2,5\sqrt{6}/2,5\sqrt{2})$. Note that this setup is similar to those considered in~\cite{Anaya22}. Figs.~\ref{fig_multiEve_a}~and~\ref{fig_multiEve_b} demonstrate the minimum secrecy capacity among the four Eves for different security radii and distances between Alice and Bob, respectively. We find that positive secrecy capacities are unachievable for the Conventional-MRT solution adopted in~\cite{Anaya22, ZhengZhang24, YunpuZhang25}~and~\cite{ZhifengTang25}. In addition, we observe that both the Maximin and low-complexity solutions achieve better secrecy performance than that of the baseline designs, which further verifies the robustness of the proposed solutions.

\begin{figure}[t]
	\centering
	\subfigure[Worst-case secrecy capacity versus the AN transmit power at Bob.]{
		\begin{minipage}[t]{0.48\linewidth}
			\centering
			\includegraphics[width=1\columnwidth]{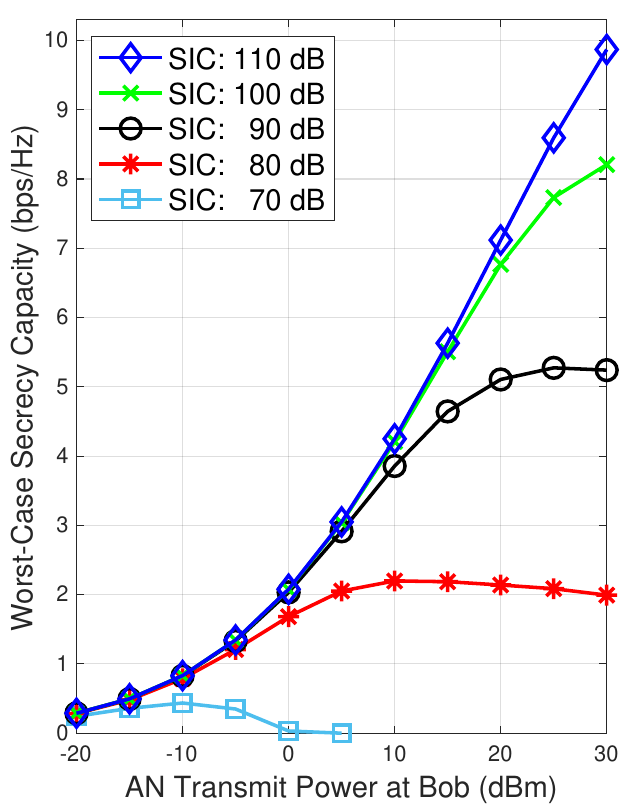}
			\label{fig_Cs_BobPower}
	\end{minipage}}
	\subfigure[Worst-case secrecy capacity versus the SIC performance of Bob.]{
		\begin{minipage}[t]{0.48\linewidth}
			\centering
			\includegraphics[width=1\columnwidth]{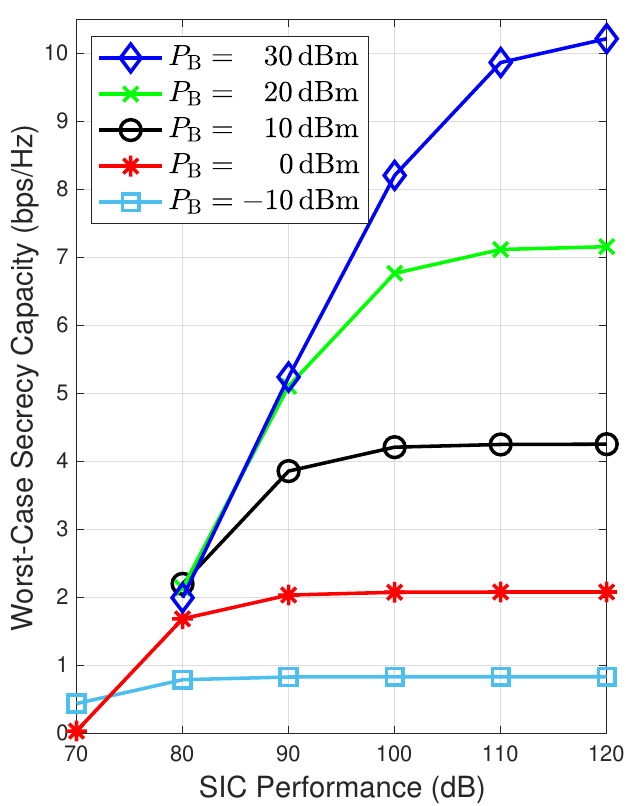}
			\label{fig_Cs_SIC}
	\end{minipage}}
	\centering
	\caption{Secrecy performance with RCVPZ for different AN transmit power and SIC performance when Bob is located at $\B{p}_{\rm B} = \left(0,0,10\right)$.}
	\label{fig_RCVPZ_performance}
\end{figure}

Finally, we observe from both Figs.~\ref{fig_radius} and \ref{fig_dist} that the secrecy performance generally becomes worse when Bob's position deviates from a direction perpendicular to Alice's UPA. This phenomenon is caused by the reduction in \textit{effective antenna array aperture}~\cite{YuanweiLiu23} as Bob moves away from the direction normal to the antenna array surface. Also, this is the reason why the performance curves of the Peak-SNR and Equal-SNRs solutions both terminate at the distance $d_{\rm AB}=10.5\,{\rm m}$ as indicated by the vertical dotted line in Figs.~\ref{fig_dist_b}~and~\ref{fig_multiEve_b}. Indeed, the global maximum SNR point cannot be located farther from Alice, due to the degraded beamfocusing performance caused by the reduced effective antenna array aperture.

\subsection{Worst-Case Secrecy Performance with RCVPZ}

\begin{figure}[t]
	\centering
	\subfigure[Optimal power allocation at Alice versus the AN transmit power at Bob.]{
		\begin{minipage}[t]{0.48\linewidth}
			\centering
			\includegraphics[width=1\columnwidth]{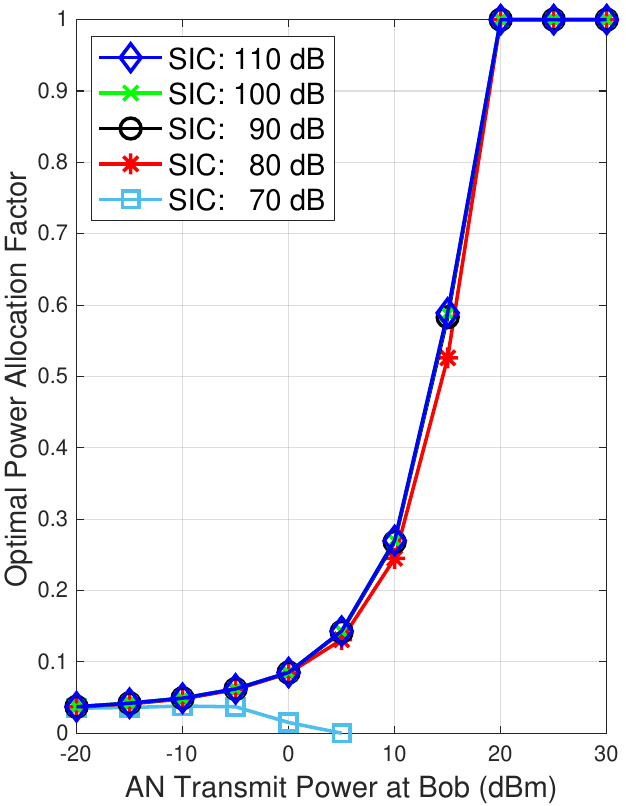}
			\label{fig_PHI_BobPower}
	\end{minipage}}
	\subfigure[Virtual security radius versus the SIC performance of Bob.]{
		\begin{minipage}[t]{0.48\linewidth}
			\centering
			\includegraphics[width=1\columnwidth]{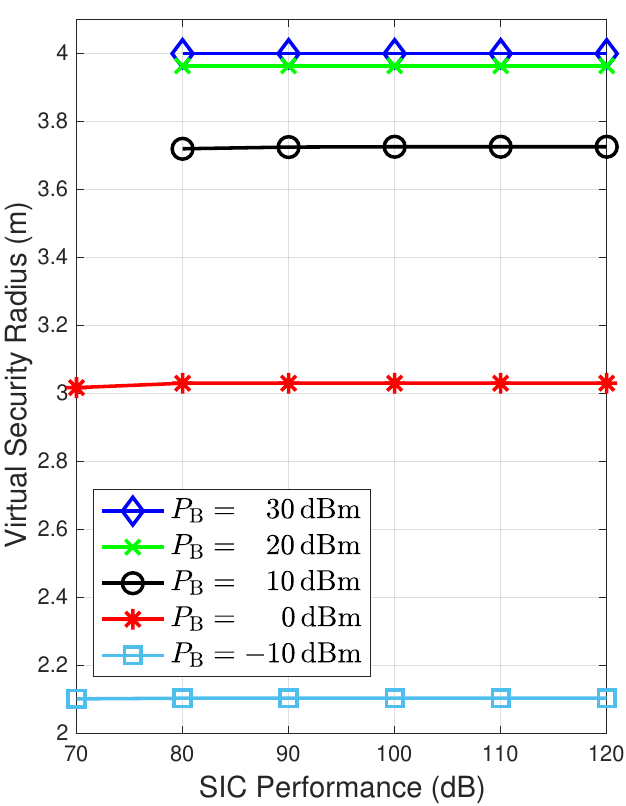}
			\label{fig_VirtualRadius_SIC}
	\end{minipage}}
	\centering
	\caption{Optimal power allocation strategy for Alice and the size of RCVPZ enabled by a full-duplex Bob for different AN transmit power and SIC performance when Bob is located at $\B{p}_{\rm B} = \left(0,0,10\right)$.}
	\label{fig_RCVPZ_behavior}
\end{figure}

In this subsection, we present numerical results for the extended SGDA-based design to analyze secrecy performance with a full-duplex Bob.

Figs.~\ref{fig_Cs_BobPower}~and~\ref{fig_Cs_SIC} depict the worst-case secrecy capacity versus the AN power and the SIC performance at Bob, respectively. In Fig.~\ref{fig_Cs_BobPower}, we observe that the worst-case secrecy capacity does not always increase with the AN power at Bob, as the optimal AN power depends on Bob's SIC performance, e.g., $P_{\rm B}^\star = 10\, {\rm dBm}$ for a SIC performance of $80\, {\rm dB}$ and $P_{\rm B}^\star = 25\, {\rm dBm}$ for a SIC performance of $90\, {\rm dB}$. This is due to the fact that as the AN power at Bob grows, the deterioration caused by undesired residual self-interference gradually outweighs the benefit of the intentional interference at the eavesdropper. In Fig.~\ref{fig_Cs_SIC}, we find that the worst-case secrecy capacity continually increases with SIC performance until saturation, e.g., when $P_{\rm B} = 0\, {\rm dBm}$, we see that $C_{\rm S} = 2.1\, {\rm bps/Hz}$ is the highest secrecy performance that can be reached when SIC suppresses the self-interference below the noise floor.

In Fig.~\ref{fig_RCVPZ_behavior}, we illustrate the optimal power allocation strategy at Alice as well as the virtual security radius for different AN transmit powers and SIC performance at Bob. In Fig.~\ref{fig_PHI_BobPower}, when the AN transmit power at Bob is very small, Alice should deploy an overwhelmingly large proportion of transmit power for AN to ensure adequate interference at Eve, e.g., the optimal power allocation factor is around $5\%$ when $P_{\rm B}$ is less than $-10\, {\rm dBm}$. As the AN transmit power at Bob increases beyond $20\, {\rm dBm}$, the optimal strategy at Alice quickly approaches one of the two extreme cases: (i) Alice either allocates all transmit power to the information signal if Bob has relatively good SIC, e.g., no less than $80\, {\rm dB}$; or (ii) Alice terminates transmission if Bob's SIC is poor, e.g., $70\, {\rm dB}$ or less, because no positive secrecy capacity can be achieved even with an optimized design. The primary reason for the two extreme cases is the different levels of residual self-interference that results from Bob's SIC, while the interference received at Eve from Bob is not impacted by Bob's SIC performance. When the AN transmit power at Bob is large, good SIC performance means that Bob can maintain a fairly low level of residual self-interference while Eve suffers from strong interference. Since Bob already has a significant advantage over Eve in terms of received interference, there is no need, and in fact it is inefficient, for Alice to also transmit AN. However, poor SIC performance at Bob leads to strong self-interference at Bob that exceeds the AN received by Eve. When the self-interference at Bob is strong, it eventually becomes impossible for Alice to sufficiently increase the secrecy capacity by allocating more power to AN; hence, a positive secrecy capacity becomes unachievable. In Fig.~\ref{fig_VirtualRadius_SIC}, we observe that the virtual security radius increases with the AN power at Bob but is insensitive to Bob's SIC performance. This is because the worst-case eavesdropping location only depends on the SINR at Eve but not that at Bob. We also note that the law of diminishing returns applies to increasing the virtual security radius when $P_{\rm B}\geq20\, {\rm dBm}$. Last but not least, we find that a positive secrecy capacity may not be achievable when Bob's SIC is poor. Therefore, in Figs.~\ref{fig_RCVPZ_performance}~and~\ref{fig_RCVPZ_behavior}, we do not show the numerical results at operating conditions where positive secrecy capacity is unachievable.

\section{Conclusion} \label{section_conclusion}

In this work, we theoretically examined how an RCPZ can be exploited in conjunction with beamfocusing design to enable near-field secure communications. An NP-hard max-min optimization problem was formulated for the worst-case secrecy capacity, to which the approximate global optimal solution was obtained using an SGDA framework. Furthermore, we proposed low-complexity beamfocusing solutions that demonstrate good secrecy performance. Additionally, we showed how an RCVPZ can be created and exploited by a full-duplex receiver to enable near-field secure communications and highlighted the impact of the SIC performance.

There are several interesting future research directions related to RCPZ-empowered near-field physical-layer security. First, given the inevitably imperfect channel estimation and/or inaccurate localization of the legitimate receiver, the beamfocusing location selectivity can be a double-edged sword that may jeopardize the secrecy performance. Hence, robust designs such as beam broadening are desirable for reducing the performance loss, enhancing the system's immunity against such uncertainties. Second, the secrecy performance with RCPZ can be further enhanced by optimizing the spatial distribution of the AN, but this is challenging when large-scale antenna arrays are employed, due to the correspondingly large AN covariance matrix. As such, efficient algorithm design is needed to handle the resultant large-scale optimization problems. Third, in multiuser communications, spatially adjacent near-field legitimate receivers may be treated as a single group protected by a common RCPZ, which simplifies the system design. In such cases, multiuser fairness can be guaranteed by proper design of the user groups, the power allocated to the groups, and the multiuser beamfocusing. We note that the proposed SGDA framework and low-complexity solutions can serve as baseline schemes for the special case where each legitimate receiver has its own individual RCPZ. Fourth, in hybrid near- and far-field communications, the co-design of an RCPZ and a conventional transmitter-side protected zone is worth investigating, as it can simultaneously achieve security for both near- and far-field receivers.

\appendix

\section{}

\subsection{Proof of Proposition \ref{proposition_1}} \label{appendix_A}

Suppose that for the given feasible $\left\{\bar\phi, \B{\bar p}_{\rm F}, \B{\bar p}_{\rm E}\right\}$, there exists a $P_{\rm A}^\prime\in\left[0,P_{\rm TX}\right]$ such that $C_{\rm S} (P_{\rm A}^\prime, \bar\phi, \B{\bar p}_{\rm F}, \B{\bar p}_{\rm E}) > 0$. Otherwise, $C_{\rm S} (P_{\rm A}, \bar\phi, \B{\bar p}_{\rm F}, \B{\bar p}_{\rm E}) = 0$ holds for the given $\left\{\bar\phi, \B{\bar p}_{\rm F}, \B{\bar p}_{\rm E}\right\}$ and any $P_{\rm A}\in\left[0,P_{\rm TX}\right]$, which is a trivial case indicating that the maximization can be achieved by $P_{\rm A}^{\rm max} = P_{\rm TX}$.

For the given $\left\{\bar\phi, \B{\bar p}_{\rm F}, \B{\bar p}_{\rm E}\right\}$, the secrecy capacity in (P2) can be written as a function of $P_{\rm A}$ as
\begin{align}
	C_{\rm S} = \log_2\left(1+\frac{aP_{\rm A}}{\sigma_{\rm B}^2}\right) - \log_2\left(1+\frac{bP_{\rm A}}{cP_{\rm A}+\sigma_{\rm E}^2}\right),
\end{align}
where $a = \bar\phi|\B{h}_{\rm B}\CT \B{\bar h}_{\rm F}|^2 > 0$, $b = \bar\phi|\B{h}_{\rm E}\CT \B{\bar h}_{\rm F}|^2 \geq 0$, and $c = \frac{1-\bar\phi}{N-1}\|\B{h}_{\rm E}\CT\B{Z}\|^2 \geq 0$ are non-negative real numbers. The derivative of the secrecy capacity with respect to the transmit power at Alice can be calculated by $\partial C_{\rm S} / \partial P_{\rm A} = f/g$ with
\begin{align}
	g = \ln2 \left(aP_{\rm A}+\sigma_{\rm B}^2\right) \left(cP_{\rm A}+\sigma_{\rm E}^2\right) \left[\left(b+c\right)P_{\rm A}+\sigma_{\rm E}^2\right],
\end{align}
the value of which is always positive for any $P_{\rm A}\in\left[0,P_{\rm TX}\right]$.

On the one hand, when $c>0$,
\begin{align}
	f = ac(b+c)P_{\rm A}^2 + 2ac\sigma_{\rm E}^2P_{\rm A} + \sigma_{\rm E}^2(a\sigma_{\rm E}^2-b\sigma_{\rm B}^2)
\end{align}
is a convex quadratic function of $P_{\rm A}$ with axis of symmetry $P_{\rm A}^{\rm s} = -\sigma_{\rm E}^2 / (b+c) < 0$. When $b/a \leq \sigma_{\rm E}^2/\sigma_{\rm B}^2$, we have $f\geq0$ for any $P_{\rm A}\in\left[0,P_{\rm TX}\right]$. Then $C_{\rm S}$ is a non-decreasing function of $P_{\rm A}$, such that the maximum is achieved by $P_{\rm A}^{\rm max} = P_{\rm TX}$. When $b/a > \sigma_{\rm E}^2/\sigma_{\rm B}^2$, we observe that $C_{\rm S}$ first decreases and then increases in $\left[0, P_{\rm TX}\right]$. Thus, the maximum secrecy capacity is achieved by either $0$ or $P_{\rm TX}$. Since $C_{\rm S}(P_{\rm TX}) \geq C_{\rm S}(P_{\rm A}^\prime) > 0 = C_{\rm S}(0)$, the maximum is achieved by $P_{\rm A}^{\rm max}=P_{\rm TX}$.

On the other hand, when $c=0$,
\begin{align}
	f = \sigma_{\rm E}^2(a\sigma_{\rm E}^2-b\sigma_{\rm B}^2)
\end{align}
is a constant function. We have $f>0$ following the assumption of $C_{\rm S}(P_{\rm A}^\prime) > 0$. Then $C_{\rm S}$ is an increasing function of $P_{\rm A}$, such that the maximum is achieved by $P_{\rm A}^{\rm max} = P_{\rm TX}$.

This completes the proof.

\begin{figure*}[b]
	\hrulefill
	\begin{align}
		\nabla_{\B{p}_{\rm E}} \check{C}_{\rm S}
		= \frac{1}{\ln2} \frac{\left|\B{h}_{\rm E}\CT \B{w}_{\rm F}\right|^2 \left[\frac{\left(1-\phi\right)P_{\rm A}}{N-1}\nabla_{\B{p}_{\rm E}}\|\B{h}_{\rm E}\CT\B{Z}\|^2 + \frac{P_{\rm B}}{2\kappa^2d_{\rm BE}^4}\left(\B{p}_{\rm B}-\B{p}_{\rm E}\right) \right] - \left[\frac{\left(1-\phi\right)P_{\rm A}}{N-1}\|\B{h}_{\rm E}\CT\B{Z}\|^2 + \frac{P_{\rm B}}{4\kappa^2 d_{\rm BE}^2} + \sigma_{\rm E}^2\right] \nabla_{\B{p}_{\rm E}}\left|\B{h}_{\rm E}\CT \B{w}_{\rm F}\right|^2}
		{\left[\frac{\left(1-\phi\right)P_{\rm A}}{N-1}\|\B{h}_{\rm E}\CT\B{Z}\|^2 + \frac{P_{\rm B}}{4\kappa^2 d_{\rm BE}^2} + \sigma_{\rm E}^2\right] \left[\left|\B{h}_{\rm E}\CT \B{w}_{\rm F}\right|^2 + \frac{\left(1-\phi\right)P_{\rm A}}{N-1}\|\B{h}_{\rm E}\CT\B{Z}\|^2 + \frac{P_{\rm B}}{4\kappa^2 d_{\rm BE}^2} + \sigma_{\rm E}^2\right]}
		\tag{54}
		\label{grad_SR_E}
	\end{align}
	
	\begin{align}
		\frac{1}{c} \min_{\B{p}_{\rm E}^\prime\in\mathcal{P}_{\rm S}} C_{\rm S} \left(\phi, \B{p}_{\rm F}, \B{p}_{\rm E}^\prime\right)
		\overset{\text{(a)}}{\leq}
		C_{\rm S} \left(\hat\phi^{\rm max}, \hat{\B{p}}_{\rm F}^{\rm max}, \hat{\B{p}}_{\rm E}^{\rm min}\right)
		\overset{\text{(b)}}{\leq}
		c\, C_{\rm S} \left(\hat\phi^{\rm max}, \hat{\B{p}}_{\rm F}^{\rm max}, \B{p}_{\rm E}\right),\
		\forall\left(\phi, \B{p}_{\rm F}, \B{p}_{\rm E}\right)\in\left[0,1\right]\times\mathcal{R}_{\rm S}\times\mathcal{P}_{\rm S}
		\tag{61}
		\label{approx_global_optimality}
	\end{align}
	
\end{figure*}

\subsection{Proof of Proposition \ref{proposition_2}} \label{appendix_B}

For the given $\left\{\B{\bar p}_{\rm F}, \B{\bar p}_{\rm E}\right\}$, the secrecy capacity in (P3) can be written as a function of $\phi$ as
\begin{align}
	C_{\rm S} = \log_2\left(1+\frac{\alpha\phi}{\sigma_{\rm B}^2}\right) - \log_2\left(1+\frac{\beta\phi}{\gamma\left(1-\phi\right)+\sigma_{\rm E}^2}\right),
\end{align}
where $\alpha = P_{\rm A}\left|\B{h}_{\rm B}\CT \B{\bar h}_{\rm F}\right|^2 > 0$, $\beta = P_{\rm A}\left|\B{h}_{\rm E}\CT \B{\bar h}_{\rm F}\right|^2 \geq 0$, and $\gamma = \frac{P_{\rm A}}{N-1} \|\B{h}_{\rm E}\CT\B{Z}\|^2 \geq 0$ are non-negative real numbers. The derivative of the secrecy capacity with respect to the power allocation factor can be calculated by $\partial C_{\rm S} / \partial \phi = h/l$ with
\begin{align}
	l = \ln2 \left(\alpha\phi+\sigma_{\rm B}^2\right) \left[\gamma\left(1-\phi\right)+\sigma_{\rm E}^2\right] \left[\gamma\left(1-\phi\right)+\beta+\sigma_{\rm E}^2\right],
\end{align}
which is always positive for any $\phi\in\left[0,1\right]$.

First, we consider $\beta\neq\gamma>0$, in which case
\begin{align}
	h = \alpha\gamma\left(\gamma-\beta\right)\phi^2 - 2\alpha\gamma\left(\gamma+\sigma_{\rm E}^2\right)\phi\notag\\
	+ \left(\gamma+\sigma_{\rm E}^2\right) \left[\alpha\left(\gamma+\sigma_{\rm E}^2\right) - \beta\sigma_{\rm B}^2\right]
\end{align}
is a quadratic function of $\phi$, whose axis of symmetry is $\phi_{\rm s} = \left(\gamma+\sigma_{\rm E}^2\right) / \left(\gamma-\beta\right)$ and whose discriminant is $\Delta = 4\alpha\beta\gamma\left(\gamma+\sigma_{\rm E}^2\right) \left[\left(\gamma-\beta\right)\sigma_{\rm B}^2+\alpha\left(\gamma+\sigma_{\rm E}^2\right)\right]$. We investigate the following two subcases:

i) $\beta<\gamma:$ In this subcase, the parabola $h$ is convex with axis of symmetry $\phi_{\rm s}>1$. On the one hand, when $\beta=0$, $h$ has a double root $\phi_{\rm s}$ resulting from $\Delta=0$, which indicates that $h>0$ for any $\phi\in\left[0,1\right]$. Then $C_{\rm S}$ is an increasing function of $\phi$, such that the maximum is achieved by $\phi^{\rm max} = 1$. On the other hand, when $\beta>0$, $h$ has two distinct real roots $\phi_a = \phi_{\rm s} - \frac{\sqrt{\Delta}}{2\alpha\gamma\left(\gamma-\beta\right)} < \phi_a^\prime = \phi_{\rm s} + \frac{\sqrt{\Delta}}{2\alpha\gamma\left(\gamma-\beta\right)}$ resulting from $\Delta>0$, which indicates that $h>0$ for $\phi\in\left(0,\min\left\{\phi_a^+,1\right\}\right)$ and $h<0$ for $\phi\in\left(\min\left\{\phi_a^+,1\right\},1\right)$. This shows that $C_{\rm S}$ first increases and then decreases in $\left[0, 1\right]$, such that the maximum is achieved by $\phi^{\rm max} = \min\left\{\phi_a^+,1\right\}$.
	
ii) $\beta>\gamma:$ In this subcase, the parabola $h$ is concave with axis of symmetry $\phi_{\rm s}<0$. When $\Delta\leq0$, we have $h<0$ for any $\phi\in\left[0,1\right]$, which implies that $C_{\rm S}$ is a decreasing function of $\phi$, such that the maximum is achieved by $\phi^{\rm max} = 0$. On the other hand, when $\Delta>0$, $h$ has two distinct real roots $\phi_a > \phi_a^\prime$, which indicates that $h>0$ for $\phi\in\left(0,\min\left\{\phi_a^+,1\right\}\right)$ and $h<0$ for $\phi\in\left(\min\left\{\phi_a^+,1\right\},1\right)$. Thus, $C_{\rm S}$ first increases and then decreases in $\left[0, 1\right]$, such that the maximum is achieved by $\phi^{\rm max} = \min\left\{\phi_a^+,1\right\}$.

Second, we consider $\beta=\gamma>0$, in which case
\begin{align}
	h = -2\alpha\beta\left(\beta+\sigma_{\rm E}^2\right)\phi + \left(\beta+\sigma_{\rm E}^2\right) \left[\alpha\left(\beta+\sigma_{\rm E}^2\right) - \beta\sigma_{\rm B}^2\right]
\end{align}
is a monotonically decreasing linear function with root $\phi_b = \frac{1}{2}+\frac{\alpha\sigma_{\rm E}^2-\beta\sigma_{\rm B}^2}{2\alpha\beta}$, which indicates that $h>0$ for $\phi\in\left(0,\min\left\{\phi_b^+,1\right\}\right)$ and $h<0$ for $\phi\in\left(\min\left\{\phi_b^+,1\right\},1\right)$. This implies that $C_{\rm S}$ first increases and then decreases in $\left[0, 1\right]$, such that the maximum is achieved by $\phi^{\rm max} = \min\left\{\phi_b^+,1\right\}$.

Last, we consider $\gamma=0$, in which case
\begin{align}
	h = \sigma_{\rm E}^2 \left(\alpha\sigma_{\rm E}^2 - \beta\sigma_{\rm B}^2\right)
\end{align}
is a constant function. Therefore, the maximum is achieved by $\phi^{\rm max} = H\left(\alpha\sigma_{\rm E}^2 - \beta\sigma_{\rm B}^2\right)$, where $H\left(\cdot\right)$ is the Heaviside step function defined by
\begin{align}
	H\left(x\right) = 
	\left\{
	\begin{aligned}
		&1, &x \geq 0,\\
		&0, &x < 0.
	\end{aligned}
	\right.
\end{align}
This completes the proof.

\subsection{Gradient Expressions in the Proposed SGDA Framework} \label{appendix_C}

In this appendix, we provide gradient expressions for the most general forms of $\nabla_{\B{p}_{\rm F}} \check{C}_{\rm S}$ and $\nabla_{\B{p}_{\rm E}} \mathcal{\check L}_\varrho$ mentioned in Section~\ref{section_extended_SGDA}. Setting $P_{\rm B}=0$ in the provided expressions below, the gradients $\nabla_{\B{p}_{\rm F}} C_{\rm S}$ and $\nabla_{\B{p}_{\rm E}} \mathcal{L}_\varrho$ used in Algorithm~\ref{alg_SGDA} can be obtained. On the other hand, when $\phi=1$, the gradients for the special case where no AN signal is transmitted by Alice can be obtained. 

First, the gradient of the secrecy capacity with respect to the FP position can be calculated by
\begin{align}
	&\nabla_{\B{p}_{\rm F}} \check{C}_{\rm S} = \frac{1}{\ln2} \Bigg[\frac{\nabla_{\B{p}_{\rm F}}\left|\B{h}_{\rm B}\CT \B{w}_{\rm F}\right|^2}{\left|\B{h}_{\rm B}\CT \B{w}_{\rm F}\right|^2 + \rho P_{\rm B} + \sigma_{\rm B}^2}\notag\\
	&- \frac{\nabla_{\B{p}_{\rm F}}\left|\B{h}_{\rm E}\CT \B{w}_{\rm F}\right|^2}{\left|\B{h}_{\rm E}\CT \B{w}_{\rm F}\right|^2 + \frac{\left(1-\phi\right)P_{\rm A}}{N-1}\|\B{h}_{\rm E}\CT\B{Z}\|^2 + \frac{P_{\rm B}}{4\kappa^2 d_{\rm BE}^2} + \sigma_{\rm E}^2}\Bigg],
	\tag{51}
	\label{grad_Cs_pF}
\end{align}
where $\nabla_{\B{p}_{\rm F}}\left|\B{h}_{\rm I}\CT \B{w}_{\rm F}\right|^2\in\mathbb{R}^3,{\rm I}\in\left\{{\rm B},{\rm E}\right\}$, is expressed as
\begin{align}
	\nabla_{\B{p}_{\rm F}}\left|\B{h}_{\rm I}\CT \B{w}_{\rm F}\right|^2 = 2\kappa\phi P_{\rm A} \sum_{k=1}^{N} \Re\Bigg( jh_{\rm I}^k\bar{h}_{\rm F}^{k*}\sum_{\substack{i=1\\i\neq k}}^{N}h_{\rm I}^{i*}\bar{h}_{\rm F}^{i} \Bigg) \frac{\B{p}_{\rm F}-\B{p}_{\rm A}^k}{d_{\rm F}^k}.
	\tag{52}
\end{align}

Second, the gradient of the AL function with respect to the position of the eavesdropper can be calculated by
\begin{align}
	\nabla_{\B{p}_{\rm E}} \mathcal{\check L}_\varrho = \nabla_{\B{p}_{\rm E}} \check{C}_{\rm S} + 2\varrho\sum_{k=0}^{N} \B{g}_k,
	\tag{53}
	\label{grad_AL}
\end{align}
where the gradient of the secrecy capacity with respect to the position of the eavesdropper is given by (\ref{grad_SR_E}), at the bottom of this page. In (\ref{grad_SR_E}), we have

\begin{align}
	\nabla_{\B{p}_{\rm E}}\left|\B{h}_{\rm E}\CT \B{w}_{\rm F}\right|^2
	= \sum_{k=1}^{N} \Bigg[ &2\Re\Bigg(\left( j\kappa - \frac{1}{d_{\rm E}^k} \right) w_{\rm F}^k h_{\rm E}^{k*}\sum_{\substack{i=1\\i\neq k}}^{N}w_{\rm F}^{i*} h_{\rm E}^i \Bigg)\notag\\
	&- \frac{\left|w_{\rm F}^k\right|^2}{2\kappa^2 {d_{\rm E}^k}^3} \Bigg] \frac{\B{p}_{\rm E}-\B{p}_{\rm A}^k}{d_{\rm E}^k},
	\tag{55}
\end{align}
and $\nabla_{\B{p}_{\rm E}}\|\B{h}_{\rm E}\CT\B{Z}\|^2 = \sum_{m=1}^{N-1} \nabla_{\B{p}_{\rm E}}\left|\B{h}_{\rm E}\CT\B{z}_m\right|^2$ with
\begin{align}
	\nabla_{\B{p}_{\rm E}}\left|\B{h}_{\rm E}\CT\B{z}_m\right|^2
	= \sum_{k=1}^{N} \Bigg[ &2\Re\Bigg(\left( j\kappa - \frac{1}{d_{\rm E}^k} \right) z_m^k h_{\rm E}^{k*}\sum_{\substack{i=1\\i\neq k}}^{N}z_m^{i*} h_{\rm E}^i \Bigg)\notag\\
	&- \frac{\left|z_m^k\right|^2}{2\kappa^2 {d_{\rm E}^k}^3} \Bigg] \frac{\B{p}_{\rm E}-\B{p}_{\rm A}^k}{d_{\rm E}^k}.
	\tag{56}
\end{align}
In addition, $\B{g}_k\in\mathbb{R}^3, k=\left\{0,1,\ldots,N\right\}$ are given by
\begin{align}
	\B{g}_0 &= \left( R_{\rm S}^2 - \left\|\B{p}_{\rm E}-\B{p}_{\rm B}\right\|^2 + \frac{\mu_0}{\varrho} \right)^+\left(\B{p}_{\rm B}-\B{p}_{\rm E}\right), \tag{57}\\
	\B{g}_k &= \left( d_{\rm R}^2 - \left\|\B{p}_{\rm E}-\B{p}_{\rm A}^k\right\|^2 + \frac{\mu_k}{\varrho} \right)^+\left(\B{p}_{\rm A}^k-\B{p}_{\rm E}\right).\tag{58}
\end{align}

\subsection{Analysis of the Approximate Global Optimality of the Maximin Solution} \label{appendix_D}

Suppose that the global solution to the max-min problem in (P3) is given by $\left(\phi^{\rm max}, \B{p}_{\rm F}^{\rm max},\B{p}_{\rm E}^{\rm min}\right)$. For the inner minimization subproblem, we have the following inequality
\begin{align}
	& C_{\rm S} \left(\phi^{\rm max}, \B{p}_{\rm F}^{\rm max}, \B{p}_{\rm E}^{\rm min}\right)
	=\min_{\B{p}_{\rm E}\in\mathcal{P}_{\rm E}} C_{\rm S} \left(\phi^{\rm max}, \B{p}_{\rm F}^{\rm max}, \B{p}_{\rm E}\right)\notag\\
	& \leq C_{\rm S} \left(\phi^{\rm max}, \B{p}_{\rm F}^{\rm max}, \B{p}_{\rm E}\right),\ \forall\B{p}_{\rm E}\in\mathcal{P}_{\rm E},
	\tag{59}
	\label{RHS}
\end{align}
where $\mathcal{P}_{\rm E} = \{\B{p}_{\rm E}\in\mathbb{R}^3| \|\B{p}_{\rm E}-\B{p}_{\rm B}\|\geq R_{\rm S}, \|\B{p}_{\rm E}-\B{p}_{\rm A}^k\| \geq d_{\rm R}, k=1,2,\ldots,N\}$ is the feasible region of the eavesdropping position. For the outer maximization subproblem, we have the following inequality
\begin{align}
	& C_{\rm S} \left(\phi^{\rm max}, \B{p}_{\rm F}^{\rm max}, \B{p}_{\rm E}^{\rm min}\right)
	= \max_{\phi, \B{p}_{\rm F}}\ \min_{\B{p}_{\rm E}^\prime}\ C_{\rm S} \left(\phi, \B{p}_{\rm F}, \B{p}_{\rm E}^\prime\right)\notag\\
	&\geq\min_{\B{p}_{\rm E}^\prime\in\mathcal{P}_{\rm E}} C_{\rm S} \left(\phi, \B{p}_{\rm F}, \B{p}_{\rm E}^\prime\right),\
	\forall \left(\phi, \B{p}_{\rm F}\right)\in\left[0,1\right]\times\mathcal{R}_{\rm B}.
	\tag{60}
	\label{LHS}
\end{align}
We clarify that (\ref{RHS}) and (\ref{LHS}) jointly define the global optimality condition for the max-min problem in (P3)~\cite{ChiJin20}.

Since finding a globally optimal solution that satisfies the conditions in (\ref{RHS}) and (\ref{LHS}) is NP-hard, the SGDA framework provides the Maximin solution $(\hat\phi^{\rm max},\hat{\B{p}}_{\rm F}^{\rm max},\hat{\B{p}}_{\rm E}^{\rm min})$ satisfying the approximate relaxed global optimality condition defined in (\ref{approx_global_optimality}), at the bottom of the previous page. The approximation involves two parts: i) the conditions in (\ref{RHS}) and (\ref{LHS}) are relaxed as (\ref{approx_global_optimality}a) and (\ref{approx_global_optimality}b) via an \textit{approximation ratio} $c>1$~\cite{TheDesignofApproxAlgs}, and ii) feasible regions $\mathcal{R}_{\rm B}$ and $\mathcal{P}_{\rm E}$ are respectively sampled to create the corresponding finite sets $\mathcal{R}_{\rm S}$ and $\mathcal{P}_{\rm S}$. The smaller the approximation ratio, the closer the Maximin solution is to global optimality.

Let $c = \max\{c_1,c_2\}$, where $c_1$ and $c_2$ are the approximation ratios of the maximization and minimization subproblems, respectively expressed as:
\begin{align}
	c_1 = \frac{\max_{\B{p}_{\rm F}\in\mathcal{R}_{\rm S}}\min_{\B{p}_{\rm E}\in\mathcal{P}_{\rm S}} C_{\rm S} \left(\Phi\left(\B{p}_{\rm F},\B{p}_{\rm E}\right), \B{p}_{\rm F}, \B{p}_{\rm E}\right)}{C_{\rm S} \left(\hat\phi^{\rm max}, \hat{\B{p}}_{\rm F}^{\rm max}, \hat{\B{p}}_{\rm E}^{\rm min}\right)} \tag{62}\label{approx_c1}
\end{align}
and
\begin{align}
	c_2 = \frac{C_{\rm S} \left(\hat\phi^{\rm max}, \hat{\B{p}}_{\rm F}^{\rm max}, \hat{\B{p}}_{\rm E}^{\rm min}\right)}{\min_{\B{p}_{\rm E}\in\mathcal{P}_{\rm S}} C_{\rm S} \left(\hat\phi^{\rm max}, \hat{\B{p}}_{\rm F}^{\rm max}, \B{p}_{\rm E}\right)}. \tag{63}\label{approx_c2}
\end{align}
We define $\mathcal{R}_{\rm S}$ as the set that results from uniformly sampling $\mathcal{R}_{\rm B}$ for $d_{\rm F}\in\left[0.1d_{\rm AB}, 10d_{\rm AB}\right]$, which has size $|\mathcal{R}_{\rm S}| = 10^4$. We define $\mathcal{P}_{\rm S}$ as the union of $\mathcal{E}_{{\rm S}1}$ and $\mathcal{E}_{{\rm S}2}$ with sizes $|\mathcal{E}_{{\rm S}1}| = |\mathcal{E}_{{\rm S}2}| = 10^4$ that are uniformly sampled from $\mathcal{E}_1$ and $\mathcal{E}_2$ with radii $\delta_{\rm 1} = \delta_{\rm 2} = 1{\rm\, m}$, respectively. Note that the sampling of $\mathcal{P}_{\rm E}$ is performed within $\mathcal{E}_1$ and $\mathcal{E}_2$ since the secrecy capacity for an eavesdropper located elsewhere would be no smaller, as analyzed in Section~\ref{section_SGDA}. For any given Maximin solution $(\hat\phi^{\rm max},\hat{\B{p}}_{\rm F}^{\rm max},\hat{\B{p}}_{\rm E}^{\rm min})$ generated by the SGDA framework, its approximation ratio can be evaluated by exhaustive search based on (\ref{approx_c1}) and (\ref{approx_c2}).

As an example, we evaluate the approximation ratio for Maximin solutions shown in Figs.~\ref{fig_radius}~and~\ref{fig_dist}. We find that the corresponding approximation ratio is no more than $1.15$ for all operating conditions. This result indicates that, under the considered settings, the SGDA framework offers a computational tractable way to generate high quality, approximately optimal solutions to the max-min problem in (P3).

%
%
%
%
%
%
%
%

\ifCLASSOPTIONcaptionsoff
\newpage
\fi

\bibliographystyle{IEEEtran}

\bibliography{Ref_NF_PLS}

\begin{thebibliography}{10}
\providecommand{\url}[1]{#1}
\csname url@samestyle\endcsname
\providecommand{\newblock}{\relax}
\providecommand{\bibinfo}[2]{#2}
\providecommand{\BIBentrySTDinterwordspacing}{\spaceskip=0pt\relax}
\providecommand{\BIBentryALTinterwordstretchfactor}{4}
\providecommand{\BIBentryALTinterwordspacing}{\spaceskip=\fontdimen2\font plus
\BIBentryALTinterwordstretchfactor\fontdimen3\font minus
  \fontdimen4\font\relax}
\providecommand{\BIBforeignlanguage}[2]{{%
\expandafter\ifx\csname l@#1\endcsname\relax
\typeout{** WARNING: IEEEtran.bst: No hyphenation pattern has been}%
\typeout{** loaded for the language `#1'. Using the pattern for}%
\typeout{** the default language instead.}%
\else
\language=\csname l@#1\endcsname
\fi
#2}}
\providecommand{\BIBdecl}{\relax}
\BIBdecl

\bibitem{CenLiu25}
C.~Liu, X.~Zhou, N.~Yang, S.~Durrani, and A.~L. Swindlehurst, ``Near-field
  beamfocusing for secure transmission with receiver-centered protected zone,''
  in \emph{Proc. IEEE Int. Conf. Commun. (ICC)}, Jun. 2025, pp. 614--619.

\bibitem{YuanweiLiu24}
Y.~Liu, C.~Ouyang, Z.~Wang, J.~Xu, X.~Mu, and A.~L. Swindlehurst, ``Near-field
  communications: A comprehensive survey,'' \emph{IEEE Commun. Surveys Tuts.},
  Jun. 2025.

\bibitem{YuanweiLiu23}
Y.~Liu, Z.~Wang, J.~Xu, C.~Ouyang, X.~Mu, and R.~Schober, ``Near-field
  communications: A tutorial review,'' \emph{IEEE Open J. Commun. Soc.},
  vol.~4, pp. 1999--2049, Aug. 2023.

\bibitem{MingyaoCui23}
M.~Cui, Z.~Wu, Y.~Lu, X.~Wei, and L.~Dai, ``Near-field {MIMO} communications
  for {6G}: Fundamentals, challenges, potentials, and future directions,''
  \emph{IEEE Commun. Mag.}, vol.~61, no.~1, pp. 40--46, Jan. 2023.

\bibitem{HaiyangZhang23}
H.~Zhang, N.~Shlezinger, F.~Guidi, D.~Dardari, and Y.~C. Eldar, ``{6G} wireless
  communications: From far-field beam steering to near-field beam focusing,''
  \emph{IEEE Commun. Mag.}, vol.~61, no.~4, pp. 72--77, Apr. 2023.

\bibitem{XiangyunZhou13}
X.~Zhou, L.~Song, and Y.~Zhang, \emph{Physical Layer Security in Wireless
  Communications}.\hskip 1em plus 0.5em minus 0.4em\relax Boca Raton, FL, USA:
  CRC Press, Nov. 2013.

\bibitem{NanYang15CM}
N.~Yang, L.~Wang, G.~Geraci, M.~Elkashlan, J.~Yuan, and M.~Di~Renzo,
  ``Safeguarding {5G} wireless communication networks using physical layer
  security,'' \emph{IEEE Commun. Mag.}, vol.~53, no.~4, pp. 20--27, Apr. 2015.

\bibitem{Mukherjee14}
A.~Mukherjee, S.~A.~A. Fakoorian, J.~Huang, and A.~L. Swindlehurst,
  ``Principles of physical layer security in multiuser wireless networks: A
  survey,'' \emph{IEEE Commun. Surveys Tuts.}, vol.~16, no.~3, pp. 1550--1573,
  Feb. 2014.

\bibitem{XiangyunZhou11}
X.~Zhou, R.~K. Ganti, J.~G. Andrews, and A.~Hjorungnes, ``On the throughput
  cost of physical layer security in decentralized wireless networks,''
  \emph{IEEE Trans. Wireless Commun.}, vol.~10, no.~8, pp. 2764--2775, Aug.
  2011.

\bibitem{Romero13}
N.~Romero-Zurita, D.~McLernon, M.~Ghogho, and A.~Swami, ``{PHY} layer security
  based on protected zone and artificial noise,'' \emph{IEEE Signal Process.
  Lett.}, vol.~20, no.~5, pp. 487--490, May. 2013.

\bibitem{Yavuz21}
Y.~Yapıcı, N.~Rupasinghe, I.~Guvenc, H.~Dai, and A.~Bhuyan, ``Physical layer
  security for {NOMA} transmission in {mmWave} drone networks,'' \emph{IEEE
  Trans. Veh. Technol.}, vol.~70, no.~4, pp. 3568--3582, Apr. 2021.

\bibitem{Anaya22}
G.~J. Anaya-López, J.~P. González-Coma, and F.~J. López-Martínez, ``Spatial
  degrees of freedom for physical layer security in {XL-MIMO},'' in \emph{Proc.
  IEEE 95th Veh. Technol. Conf. (VTC Spring)}, Jun. 2022.

\bibitem{ZhengZhang24}
Z.~Zhang, Y.~Liu, Z.~Wang, X.~Mu, and J.~Chen, ``Physical layer security in
  near-field communications,'' \emph{IEEE Trans. Veh. Technol.}, vol.~73,
  no.~7, pp. 10\,761--10\,766, Jul. 2024.

\bibitem{Nasir24}
A.~A. Nasir, ``Max-min secrecy rate optimization through beam focusing in
  near-field communications,'' \emph{IEEE Commun. Lett.}, vol.~28, no.~7, pp.
  1594--1598, Jul. 2024.

\bibitem{Droulias24}
S.~Droulias, G.~Stratidakis, and A.~Alexiou, ``Beam-focusing in near-field
  communications: Enhancing the physical layer security,'' in \emph{Proc. IEEE
  25th Int. Workshop Signal Process. Adv. Wireless Commun. (SPAWC)}, Sep. 2024,
  pp. 216--220.

\bibitem{YuchenZhang24}
Y.~Zhang, H.~Zhang, S.~Xiao, W.~Tang, and Y.~C. Eldar, ``Near-field wideband
  secure communications: An analog beamfocusing approach,'' \emph{IEEE Trans.
  Signal Process.}, vol.~72, pp. 2173--2187, Apr. 2024.

\bibitem{BoqunZhao24}
B.~Zhao, C.~Ouyang, X.~Zhang, and Y.~Liu, ``Performance analysis of physical
  layer security: From far-field to near-field,'' \emph{IEEE Trans. Wireless
  Commun.}, vol.~24, no.~9, pp. 7392--7408, Sep. 2025.

\bibitem{YunpuZhang25}
Y.~Zhang, Y.~Fang, C.~You, Y.-J.~A. Zhang, and H.~C. So, ``Performance analysis
  and low-complexity beamforming design for near-field physical layer
  security,'' \emph{arXiv preprint arXiv:2407.13491}, Apr. 2025.

\bibitem{ZhifengTang25}
Z.~Tang, N.~Yang, X.~Zhou, S.~Durrani, M.~Juntti, and J.~M. Jornet, ``Low
  complexity artificial noise aided beam focusing design in near-field
  terahertz communications,'' \emph{arXiv preprint arXiv:2502.08967}, Aug.
  2025.

\bibitem{Ferreira24}
J.~Ferreira, J.~Guerreiro, and R.~Dinis, ``Physical layer security with
  near-field beamforming,'' \emph{IEEE Access}, vol.~12, pp. 4801--4811, Jan.
  2024.

\bibitem{WeijianChen25}
W.~Chen, Z.~Wei, and Z.~Yang, ``Robust beamfocusing for secure {NFC} with
  imperfect {CSI},'' \emph{Sensors}, vol.~25, no.~4, Feb. 2025.

\bibitem{JiangongChen24}
J.~Chen, Y.~Xiao, K.~Liu, Y.~Zhong, X.~Lei, and M.~Xiao, ``Physical layer
  security for near-field communications via directional modulation,''
  \emph{IEEE Trans. Veh. Technol.}, vol.~73, no.~8, pp. 12\,242--12\,246, Aug.
  2024.

\bibitem{LinlinSun25}
L.~Sun, Z.~Chen, T.~Huang, F.~Shu, and S.~Yan, ``Near-field communication with
  random frequency diverse array,'' \emph{IEEE Wireless. Commun. Lett.},
  vol.~14, no.~1, pp. 213--217, Jan. 2025.

\bibitem{ZihaoTeng25}
Z.~Teng, J.~An, C.~Masouros, H.~Li, L.~Gan, and D.~W.~K. Ng, ``Dynamic
  precoding for near-field secure communications: Implementation and
  performance analysis,'' \emph{IEEE Internet Things J.}, vol.~12, no.~15, pp.
  29\,427--29\,442, Aug. 2025.

\bibitem{Bjornson21}
E.~Bj\"ornson, {\"O}.~T. Demir, and L.~Sanguinetti, ``A primer on near-field
  beamforming for arrays and reconfigurable intelligent surfaces,'' in
  \emph{Proc. 55th Asilomar Conf. Signals, Syst., Comput.}, Oct. 2021, pp.
  105--112.

\bibitem{Selvan17}
K.~T. Selvan and R.~Janaswamy, ``Fraunhofer and fresnel distances: Unified
  derivation for aperture antennas,'' \emph{IEEE Antennas Propag. Mag.},
  vol.~59, no.~4, pp. 12--15, Aug. 2017.

\bibitem{YuanjianLi25}
Y.~Li and A.~S. Madhukumar, ``Hybrid near- and far-field {THz UM-MIMO} channel
  estimation: A sparsifying matrix learning-aided {Bayesian} approach,''
  \emph{IEEE Trans. Wireless Commun.}, vol.~24, no.~3, pp. 1881--1897, Mar.
  2025.

\bibitem{YueWang18}
Y.~Wang and K.~C. Ho, ``Unified near-field and far-field localization for {AOA}
  and hybrid {AOA-TDOA} positionings,'' \emph{IEEE Trans. Wireless Commun.},
  vol.~17, no.~2, pp. 1242--1254, Feb. 2018.

\bibitem{Goel08}
S.~Goel and R.~Negi, ``Guaranteeing secrecy using artificial noise,''
  \emph{IEEE Trans. Wireless Commun.}, vol.~7, no.~6, pp. 2180--2189, Jun.
  2008.

\bibitem{XiangyunZhou10}
X.~Zhou and M.~R. McKay, ``Secure transmission with artificial noise over
  fading channels: Achievable rate and optimal power allocation,'' \emph{IEEE
  Trans. Veh. Technol.}, vol.~59, no.~8, pp. 3831--3842, Oct. 2010.

\bibitem{NanYang15}
N.~Yang, S.~Yan, J.~Yuan, R.~Malaney, R.~Subramanian, and I.~Land, ``Artificial
  noise: Transmission optimization in multi-input single-output wiretap
  channels,'' \emph{IEEE Trans. Commun.}, vol.~63, no.~5, pp. 1771--1783, May
  2015.

\bibitem{YuanjianLi17}
Y.~Li, R.~Zhao, X.~Tan, and Z.~Nie, ``Secrecy performance analysis of
  artificial noise aided precoding in full-duplex relay systems,'' in
  \emph{Proc. IEEE Global Commun. Conf. (GLOBECOM)}, 2017, pp. 1--6.

\bibitem{Razaviyayn20}
M.~Razaviyayn, T.~Huang, S.~Lu, M.~Nouiehed, M.~Sanjabi, and M.~Hong,
  ``Nonconvex min-max optimization: Applications, challenges, and recent
  theoretical advances,'' \emph{IEEE Signal Process. Mag.}, vol.~37, no.~5, pp.
  55--66, Sep. 2020.

\bibitem{ChiJin20}
C.~Jin, P.~Netrapalli, and M.~Jordan, ``What is local optimality in
  nonconvex-nonconcave minimax optimization?'' in \emph{Proc. 37th Int. Conf.
  Mach. Learn. (ICML)}, vol. 119, Jul. 2020, pp. 4880--4889.

\bibitem{TianyiLin20}
T.~Lin, C.~Jin, and M.~Jordan, ``On gradient descent ascent for
  nonconvex-concave minimax problems,'' in \emph{Proc. 37th Int. Conf. Mach.
  Learn. (ICML)}, vol. 119, Jul. 2020, pp. 6083--6093.

\bibitem{GPM}
J.~B. Rosen, ``The gradient projection method for nonlinear programming. {P}art
  {II}. {N}onlinear constraints,'' \emph{J. Soc. Ind. Appl. Mathematics},
  vol.~9, no.~4, pp. 514--19, Dec. 1961.

\bibitem{NumericalOptimization}
J.~Nocedal and S.~J. Wright, \emph{Numerical Optimization}.\hskip 1em plus
  0.5em minus 0.4em\relax New York, NY, USA: Springer, 2006.

\bibitem{ZaiwenWen10}
Z.~Wen, D.~Goldfarb, and W.~Yin, ``Alternating direction augmented {L}agrangian
  methods for semidefinite programming,'' \emph{Math. Prog. Comp.}, vol.~2,
  no.~3, pp. 203--230, Sep. 2010.

\bibitem{XianghaoYu16}
X.~Yu, J.-C. Shen, J.~Zhang, and K.~B. Letaief, ``Alternating minimization
  algorithms for hybrid precoding in millimeter wave {MIMO} systems,''
  \emph{IEEE J. Sel. Topics Signal Process.}, vol.~10, no.~3, pp. 485--500,
  Apr. 2016.

\bibitem{ShihaoYan18}
S.~Yan, X.~Zhou, N.~Yang, T.~D. Abhayapala, and A.~L. Swindlehurst, ``Secret
  channel training to enhance physical layer security with a full-duplex
  receiver,'' \emph{IEEE Trans. Inf. Forensics Security}, vol.~13, no.~11, pp.
  2788--2800, Nov. 2018.

\bibitem{GanZheng13}
G.~Zheng, I.~Krikidis, J.~Li, A.~P. Petropulu, and B.~Ottersten, ``Improving
  physical layer secrecy using full-duplex jamming receivers,'' \emph{IEEE
  Trans. Signal Process.}, vol.~61, no.~20, pp. 4962--4974, Oct. 2013.

\bibitem{WeiLi12}
W.~Li, M.~Ghogho, B.~Chen, and C.~Xiong, ``Secure communication via sending
  artificial noise by the receiver: Outage secrecy capacity/region analysis,''
  \emph{IEEE Commun. Lett.}, vol.~16, no.~10, pp. 1628--1631, Oct. 2012.

\bibitem{Bharadia13}
D.~Bharadia, E.~McMilin, and S.~Katti, ``Full duplex radios,'' in \emph{Proc.
  ACM SIGCOMM}, Aug. 2013, p. 375–386.

\bibitem{Sabharwal14}
A.~Sabharwal, P.~Schniter, D.~Guo, D.~W. Bliss, S.~Rangarajan, and R.~Wichman,
  ``In-band full-duplex wireless: Challenges and opportunities,'' \emph{IEEE J.
  Sel. Areas Commun.}, vol.~32, no.~9, pp. 1637--1652, Sep. 2014.

\bibitem{Korpi16}
D.~Korpi, J.~Tamminen, M.~Turunen, T.~Huusari, Y.-S. Choi, L.~Anttila,
  S.~Talwar, and M.~Valkama, ``Full-duplex mobile device: Pushing the limits,''
  \emph{IEEE Commun. Mag.}, vol.~54, no.~9, pp. 80--87, Sep. 2016.

\bibitem{HaiquanLu22}
H.~Lu and Y.~Zeng, ``Communicating with extremely large-scale array/surface:
  Unified modeling and performance analysis,'' \emph{IEEE Trans. Wireless
  Commun.}, vol.~21, no.~6, pp. 4039--4053, Jun. 2022.

\bibitem{TheDesignofApproxAlgs}
D.~P. Williamson and D.~B. Shmoys, \emph{The Design of Approximation
  Algorithms}.\hskip 1em plus 0.5em minus 0.4em\relax New York, NY, USA:
  Cambridge University Press, 2011.

\end{thebibliography}

 





\vfill

\end{document}